\documentclass{vldb}

\AtBeginDocument{%
  \providecommand\BibTeX{{%
    \normalfont B\kern-0.5em{\scshape i\kern-0.25em b}\kern-0.8em\TeX}}}

\usepackage[linesnumbered,ruled,vlined,noend]{algorithm2e}
\let\oldnl\nl
\newcommand{\nonl}{\renewcommand{\nl}{\let\nl\oldnl}}
\SetNlSty{bfseries}{\color{black}}{}

\usepackage{graphicx}
\usepackage{balance}  
\usepackage{xcolor}
\usepackage{caption}

\usepackage{amssymb,amsmath,amsthm}
\newtheorem{thm}{Theorem}
\newtheorem{lemma}{Lemma}
\usepackage{setspace}
\usepackage{multirow}

\usepackage{graphics}
\usepackage{mdwlist}
\usepackage[normalem]{ulem}

\usepackage{booktabs}
\usepackage{soul}
\usepackage[center]{subfigure}
\usepackage[export]{adjustbox}
\usepackage{bm}
\usepackage{bbm}
\usepackage{url}
\usepackage{bbold}

\usepackage{enumitem}
\usepackage{pifont}

\vldbTitle{Hypergraph Motifs: Concepts, Algorithms, and Discoveries}
\vldbAuthors{Geon Lee, Jihoon Ko, and Kijung Shin}
\vldbDOI{https://doi.org/10.14778/3407790.3407823}
\vldbVolume{13}
\vldbNumber{11}
\vldbYear{2020}

\begin{document}
\title{Hypergraph Motifs: Concepts, Algorithms, and Discoveries}

\numberofauthors{3} 

\author{
	%
	%
	\alignauthor
	Geon Lee\\
	\affaddr{KAIST AI}\\
	\email{geonlee0325@kaist.ac.kr}
	\alignauthor
	Jihoon Ko\\
	\affaddr{KAIST AI}\\
	\email{jihoonko@kaist.ac.kr}
	\alignauthor Kijung Shin \\\
	\affaddr{KAIST AI \& EE}\\
	\email{kijungs@kaist.ac.kr}
}

\newcommand\kijung[1]{\textcolor{red}{[Kijung:#1]}}
\newcommand\geon[1]{\textcolor{blue}{[Geon:#1]}}
\newcommand\red[1]{\textcolor{red}{#1}}
\newcommand\blue[1]{\textcolor{blue}{#1}}
\newcommand\change[1]{\textcolor{black}{#1}}
\newcommand\nochange[1]{\textcolor{black}{#1}}
\newcommand\Motif{H-motif\xspace}
\newcommand\Motifs{H-motifs\xspace}
\newcommand\motif{h-motif\xspace}
\newcommand\motifs{h-motifs\xspace}
\newcommand\method{\textsf{MoCHy}\xspace}
\newcommand\methodE{\textsf{MoCHy-E}\xspace}
\newcommand\methodEN{\textsf{MoCHy-E\textsubscript{ENUM}}\xspace}
\newcommand\methodA{\textsf{MoCHy-A}\xspace}
\newcommand\methodAE{\textsf{MoCHy-A}\xspace}
\newcommand\methodAW{\textsf{MoCHy-A\textsuperscript{+}}\xspace}

\newcommand\methodX{MoCHy\xspace}
\newcommand\methodEX{MoCHy-E\xspace}
\newcommand\methodENX{MoCHy-E\textsubscript{ENUM}\xspace}
\newcommand\methodAX{MoCHy-A\xspace}
\newcommand\methodAEX{MoCHy-A\xspace}
\newcommand\methodAWX{MoCHy-A\textsuperscript{+}\xspace}

\newcommand\naive{na\"ive\xspace}
\newcommand\Naive{Na\"ive\xspace}

\newcommand\hwedge{hyperwedge\xspace}
\newcommand\hwedges{hyperwedges\xspace}
\newcommand\Hwedge{Hyperwedge\xspace}
\newcommand\Hwedges{Hyperwedges\xspace}

\newcommand\MT{M[t]\xspace}
\newcommand\MB{\bar{M}\xspace}
\newcommand\MBT{\bar{M}[t]\xspace}
\newcommand\MH{\hat{M}\xspace}
\newcommand\MHT{\hat{M}[t]\xspace}
\newcommand\MD{\tilde{M}\xspace}
\newcommand\MDT{\tilde{M}[t]\xspace}

\newcommand\mB{\bar{m}\xspace}
\newcommand\mBT{\bar{m}[t]\xspace}
\newcommand\mH{\hat{m}\xspace}
\newcommand\mHT{\hat{m}[t]\xspace}

\newcommand\Mi{M_{e_i}\xspace}
\newcommand\Mij{M_{\wedge_{ij}}\xspace}
\newcommand\Mijk{M_{\sqcap_{ijk}}\xspace}
\newcommand\MTi{M_{e_i}[t]\xspace}
\newcommand\MTij{M_{\wedge_{ij}}[t]\xspace}
\newcommand\MTijk{M_{\sqcap_{ijk}}[t]\xspace}

\newcommand\wij{\wedge_{ij}\xspace}
\newcommand\wik{\wedge_{ik}\xspace}
\newcommand\wjk{\wedge_{jk}\xspace}
\newcommand\wki{\wedge_{ki}\xspace}
\newcommand\eijk{\{e_i,e_j,e_k\}\xspace}
\newcommand\hijk{h(\{e_i,e_j,e_k\})\xspace}
\newcommand{\egeneral}{\{e_{s_1},e_{s_2},...,e_{s_k}\}\xspace}
\newcommand{\hgeneral}{h(\{e_{s_1},e_{s_2},...,e_{s_k}\})\xspace}

\newcommand\GT{\bar{G}}

\newcommand\NT{{N}}
\newcommand\nei{N_{e_{i}}}
\newcommand\nej{N_{e_{j}}}
\newcommand\PT{\wedge}

\newcommand\degt[1]{|\NT_{#1}|}
\newcommand{\smallsection}[1]{{\vspace{0.02in} \noindent {\bf{\underline{\smash{#1}}}}}}

\maketitle

\begin{abstract}
 Hypergraphs naturally represent group interactions, which are omnipresent in many domains: collaborations of researchers, co-purchases of items, joint interactions of proteins, to name a few.
In this work, we propose tools for answering the following questions in a systematic manner:
{\bf (Q1)} what are structural design principles of real-world hypergraphs?
{\bf (Q2)} how can we compare local structures of hypergraphs of different sizes?
{\bf (Q3)} how can we identify domains which hypergraphs are from?
We first define {\it hypergraph motifs} (\motifs), which describe the connectivity patterns of three connected hyperedges.
Then, we define the significance of each \motif in a hypergraph as its occurrences relative to those in properly randomized hypergraphs.
Lastly, we define the {\it characteristic profile} (CP) as the vector of the normalized significance of every \motif.
Regarding Q1, we find that \motifs' occurrences in $11$ real-world hypergraphs from $5$ domains are clearly distinguished from those of randomized hypergraphs. In addition, we demonstrate that CPs capture local structural patterns unique to each domain, and thus comparing CPs of hypergraphs addresses Q2 and Q3.
Our algorithmic contribution is to propose \method, a family of parallel algorithms for counting \motifs' occurrences in a hypergraph. 
We  theoretically analyze their speed and accuracy, and we show empirically that the advanced approximate version \methodAW is up to $25\times$ more accurate and $32\times$ faster than the basic approximate and exact versions, respectively.

\end{abstract}

\vspace{-0.5mm}
\section{Introduction}
\label{sec:intro}

Complex systems consisting of pairwise interactions between individuals or objects are naturally expressed in the form of graphs. Nodes and edges, which compose a graph, represent individuals (or objects) and their pairwise interactions, respectively.
Thanks to their powerful expressiveness, graphs have been used in a wide variety of fields, including social network analysis, web, bioinformatics, and epidemiology. Global structural patterns of real-world graphs, such as power-law degree distribution \cite{barabasi1999emergence,faloutsos1999power} and six degrees of separation \cite{kang2010radius,watts1998collective}, have been extensively investigated.

\begin{figure}[t]
	\centering 
	\includegraphics[width=1.0\linewidth]{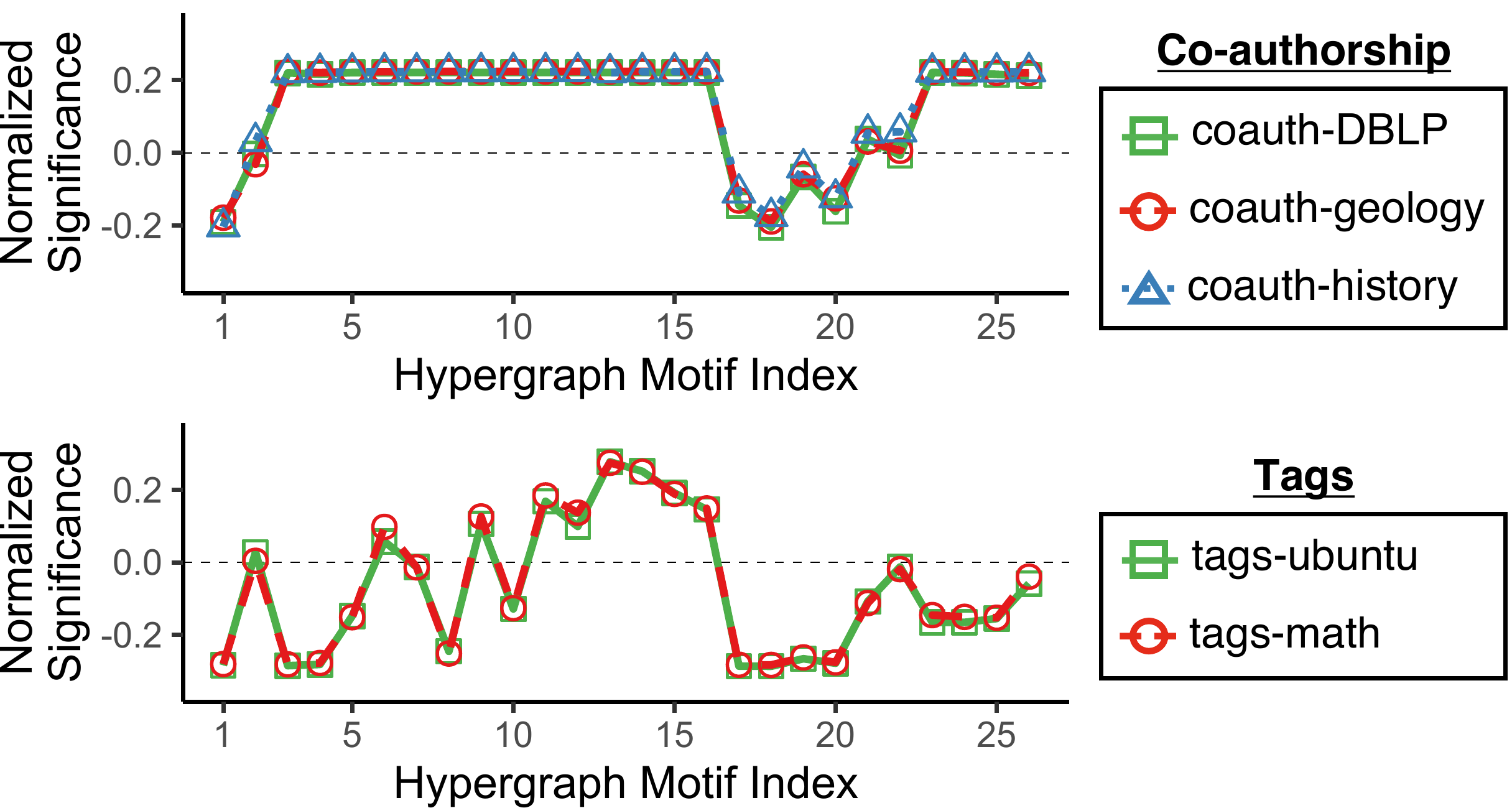}
	\caption{\label{fig:crown} 
	Distributions of \motifs' instances precisely characterize local structural patterns of real-world hypergraphs.
	Note that the hypergraphs from the same domains have similar distributions, while the hypergraphs from different domains do not.
	See Section~\ref{sec:exp:domain} for details.	
}
\end{figure}

\begin{figure*}[t]
	\centering
	\vspace{-5mm}
	\hspace{-2mm}
	\subfigure[Example data: coauthorship relations\label{fig:example:coauthorship}]{
		\includegraphics[width=0.5\columnwidth]{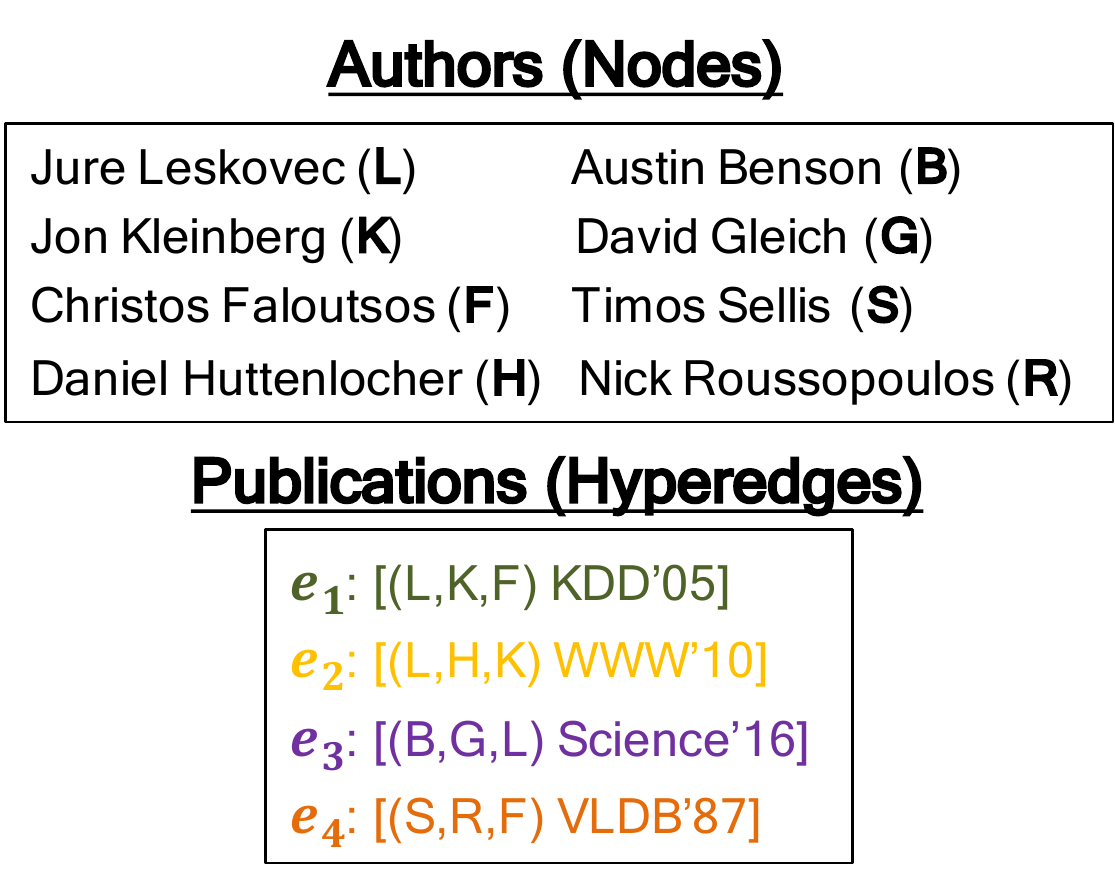}
	}
	\subfigure[Hypergraph representation\label{fig:example:hypergraph}]{
		\includegraphics[width=0.3245\columnwidth]{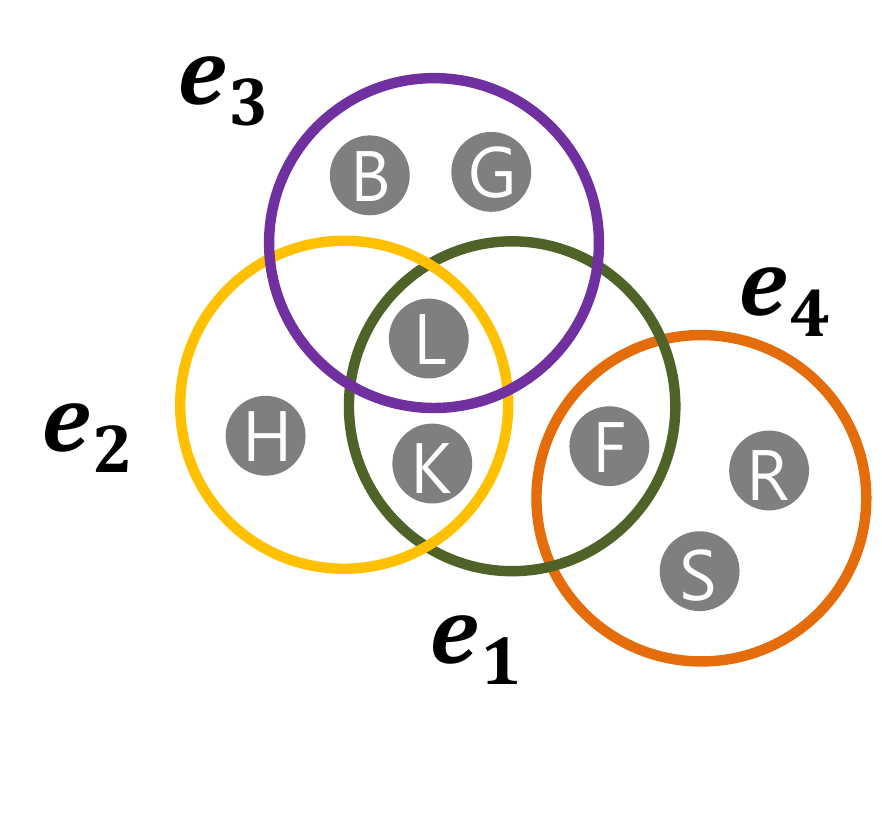}
	}
	\subfigure[Projected graph\label{fig:example:graph}]{
		\includegraphics[width=0.2135\columnwidth]{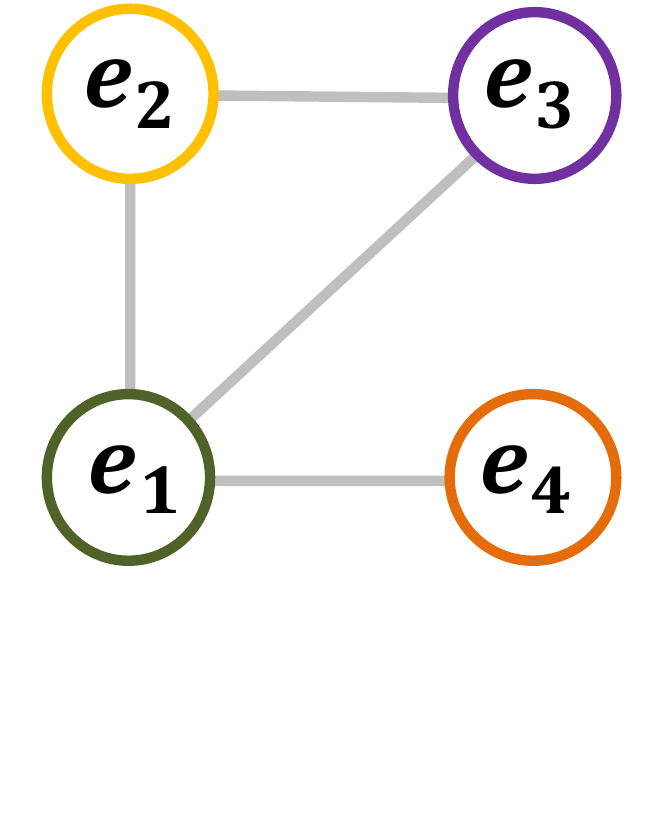} 
	}
	\subfigure[\Motifs and instances \label{fig:example:motif}]{
		\includegraphics[width=0.937\columnwidth]{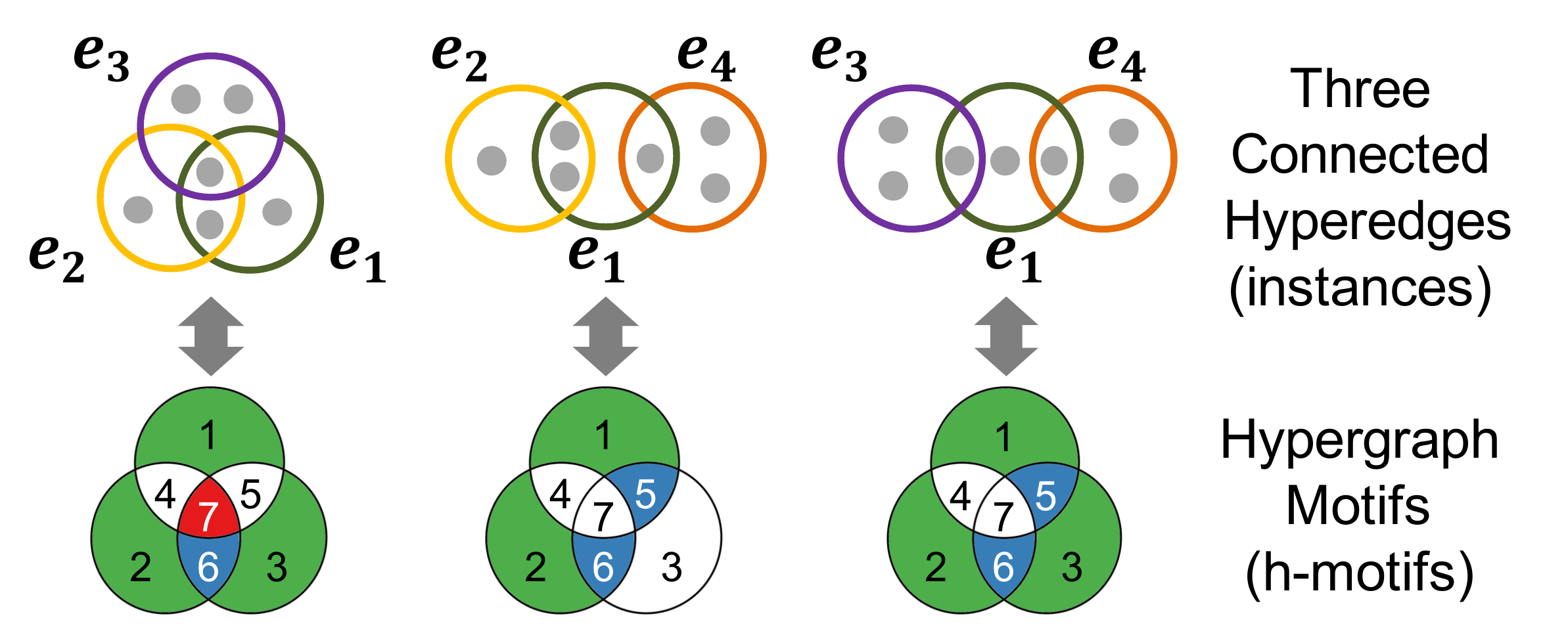}
	}
	\\
	\vspace{-2mm}
	\caption{
		\label{fig:example}
		(a) Example: co-authorship relations.
			(b) Hypergraph: the hypergraph representation of (a).
			(c) Projected Graph: the projected graph of (b).
			(d) Hypergraph Motifs: example \motifs and their instances in (b).}
\end{figure*}

In addition to global patterns, real-world graphs exhibit patterns in their local structures, which differentiate graphs in the same domain from random graphs or those in other domains. 
Local structures are revealed by counting the occurrences of different network motifs \cite{milo2004superfamilies,milo2002network}, which describe the patterns of pairwise interactions between a fixed number of connected nodes (typically $3$, $4$, or $5$ nodes).   
As a fundamental building block, network motifs have played a key role in many analytical and predictive tasks, including community detection \cite{benson2016higher,li2019edmot,tsourakakis2017scalable,yin2017local}, classification \cite{chen2013identification,lee2019graph,milo2004superfamilies}, and anomaly detection \cite{becchetti2010efficient,shin2020fast}.


Despite the prevalence of graphs, interactions in many complex systems are groupwise rather than pairwise: collaborations of researchers, co-purchases of items, joint interactions of proteins, tags attached to the same web post, to name a few. 
These group interactions cannot be represented by edges in a graph.
Suppose three or more researchers coauthor a publication. This co-authorship cannot be represented as a single edge, and creating edges between all pairs of the researchers cannot be distinguished from multiple papers coauthored by subsets of the researchers.

This inherent limitation of graphs is addressed by hypergraphs, which consist of nodes and hyperedges.
Each hyperedge is a subset of any number of nodes, and it represents a group interaction among the nodes.
For example, the coauthorship relations in Figure~\ref{fig:example:coauthorship} are naturally represented as the hypergraph in Figure~\ref{fig:example:hypergraph}.
In the hypergraph, seminar work \cite{leskovec2005graphs} coauthored by Jure Leskovec (L), Jon Kleinberg (K), and Christos Faloutsos (F) is expressed as the hyperedge $e_{1}=\{L,K,F\}$, and it is distinguished from three papers coauthored by each pair, which, if they exist, can be represented as three hyperedges $\{K,L\}$, $\{F,L\}$, and $\{F,K\}$.



\begin{table}[t!]
	\begin{center}
		\caption{\label{notations}Frequently-used symbols.}
		\scalebox{0.82}{
			\begin{tabular}{c|l}
				\toprule
				\textbf{Notation} & \textbf{Definition}\\
				\midrule
				$G=(V,E)$ & hypergraph with nodes $V$ and hyperedges $E$\\
				$E=\{e_1,...,e_{|E|}\}$ & set of hyperedges \\
				$E_v$ & set of hyperedges that contains a node $v$\\
				\midrule
				$\wedge$ & set of \hwedges in $G$\\
				$\wij$ & \hwedge consisting of $e_i$ and $e_j$ \\
				\midrule
				$\GT=(E,\wedge, \omega)$ & projected graph of $G$ \\
				$\omega(\wij)$ & the number of nodes shared between $e_i$ and $e_j$ \\
				$\nei$ & set of neighbors of $e_i$ in $\GT$ \\
				\midrule
				$h(\{e_i,e_j,e_k\})$ & \motif corresponding to an instance $\{e_i,e_j,e_k\}$ \\ 
				$M[t]$ & count of \motif $t$'s instances\\
				\bottomrule 
			\end{tabular}}
	\end{center}
	\vspace{-4mm}
\end{table}

The successful investigation and discovery of local structural patterns in real-world graphs motivate us to explore local structural patterns in real-world hypergraphs.
However, network motifs, which proved to be useful for graphs, are not trivially extended to hypergraphs.
Due to the flexibility in the size of hyperedges, there can be infinitely many patterns of interactions among a fixed number of nodes, and other nodes can also be associated with these interactions.

In this work, taking these challenges into consideration, we define $26$ {\it hypergraph motifs} (\motifs) so that they describe connectivity patterns of three connected hyperedges (rather than nodes).
As seen in Figure~\ref{fig:example:motif}, \motifs describe the connectivity pattern of hyperedges $e_{1}$, $e_{2}$, and $e_{3}$ by the emptiness of seven subsets: $e_{1}\setminus e_{2} \setminus e_{3}$, $e_{2}\setminus e_{3} \setminus e_{1}$,  $e_{3}\setminus e_{1} \setminus e_{2}$, $e_{1}\cap e_{2} \setminus e_{3}$, $e_{2}\cap e_{3} \setminus e_{1}$, $e_{3}\cap e_{1} \setminus e_{2}$, and $e_{1}\cap e_{2} \cap e_{3}$.
As a result, every connectivity pattern is described by a unique \motif, independently of the sizes of hyperedges. 
While this work focuses on connectivity patterns of three hyperedges, \motifs are easily extended to four or more hyperedges.


We count the number of each \motif's instances in $11$ real-world hypergraphs from $5$ different domains.
Then, we measure the significance of each \motif in each hypergraph by comparing the count of its instances in the hypergraph against the counts in properly randomized hypergraphs. Lastly, we compute the {\it characteristic profile} (CP) of each hypergraph, defined as the vector of the normalized significance of every \motif.
Comparing the counts and CPs of different hypergraphs leads to the following observations:
\begin{itemize}[leftmargin=*]
	\itemsep-0.2em 
	\item Structural design principles of real-world hypergraphs that are captured by frequencies of different \motifs are clearly distinguished from those of randomized hypergraphs.
	\item Hypergraphs from the same domains have similar CPs, while hypergraphs from different domains have distinct CPs (see Figure~\ref{fig:crown}). In other words, CPs successfully capture local structure patterns unique to each domain.
\end{itemize}

Our algorithmic contribution is to design \method (\textbf{Mo}tif \textbf{C}ounting in \textbf{Hy}pergraphs), a family of parallel algorithms for counting \motifs' instances, which is the computational bottleneck of the aforementioned process.
Note that since non-pairwise interactions are taken into consideration, counting the instances of \motifs is more challenging than counting the instances of network motifs, which are defined solely based on pairwise interactions.
We provide one exact version, named \methodE, and two approximate versions, named \methodAE and \methodAW.
Empirically, \methodAW is up to $25\times$ more accurate than \methodAE, and it is up to $32\times$ faster than \methodE, with little sacrifice of accuracy.
These empirical results are consistent with our theoretical analyses.

In summary, our contributions are summarized as follow:
\begin{itemize}[leftmargin=*]
	\itemsep-0.2em 
    \item {\bf Novel Concepts:} We propose \motifs, the counts of whose instances capture local structures of hypergraphs, independently of the sizes of hyperedges or hypergraphs. 
    \item {\bf Fast and Provable Algorithms:} We develop \method, a family of parallel algorithms for counting \motifs' instances. We show theoretically and empirically that the advanced version significantly outperforms the basic ones, providing a better trade-off between speed and accuracy.
    \item {\bf Discoveries in $11$ Real-world Hypergraphs:} We show that \motifs and CPs reveal local structural patterns that are shared by hypergraphs from the same domains but distinguished from those of random hypergraphs and hypergraphs from other domains (see Figure~\ref{fig:crown}).
\end{itemize}

\noindent{\bf Reproducibility:} The code and datasets used in this work are available at \url{https://github.com/geonlee0325/MoCHy}.

In Section~\ref{sec:concept}, we introduce \motifs and characteristic profiles. 
In Section~\ref{sec:method}, we present exact and approximate algorithms for counting instances of \motifs, and we analyze their theoretical properties. 
In Section~\ref{sec:exp}, we provide experimental results. 
After discussing related work in Section~\ref{sec:related}, we offer conclusions in Section~\ref{sec:summary}.

\vspace{-0.5mm}
\section{Proposed Concepts}
\label{sec:concept}

\begin{figure*}[t]
	\centering
	\vspace{-4mm}
	\includegraphics[width=1.01\linewidth]{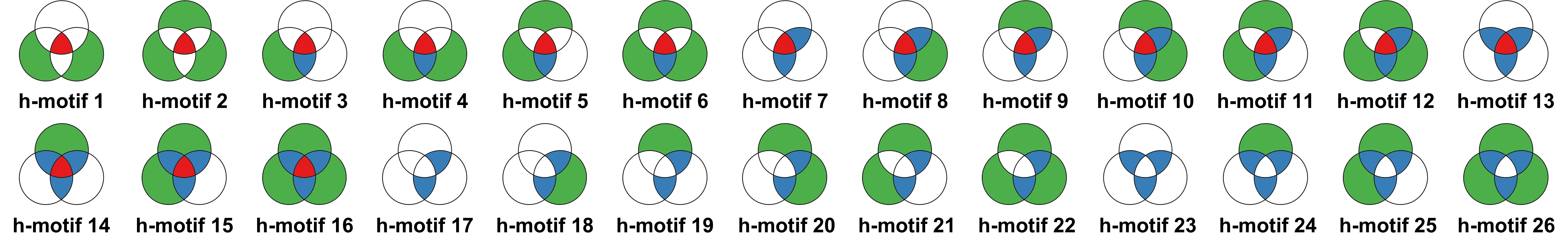} \\
	\vspace{-1mm}
	\caption{\label{motif_three_hyperedges} The 26 \motifs studied in this work. Note that
		\motifs 17 - 22 are open, while the others are closed.
	}
\end{figure*}


In this section, we introduce the proposed concepts: hypergraph motifs and characteristic profiles. Refer Table \ref{notations} for the notations frequently used throughout the paper.

\vspace{-0.5mm}
\subsection{Preliminaries and Notations}
\label{sec:concept:prelim}

We define some preliminary concepts and their notations.

\smallsection{Hypergraph} Consider a {\it hypergraph} $G=(V,E)$, where $V$ and $E:=\{e_1,e_2,...,e_{|E|}\}$ are sets of nodes and hyperedges, respectively. 
Each hyperedge $e_i\in E$ is a non-empty subset of $V$, and we use $|e_i|$ to denote the number of nodes in it.
For each node $v\in V$, we use $E_v:=\{e_i\in E: v\in e_i\}$ to denote the set of hyperedges that include $v$. 
We say two hyperedges $e_i$ and $e_j$ are {\it adjacent} if they share any member, i.e., if $e_i\cap e_j \neq \varnothing$.
Then, for each hyperedge $e_i$,
we denote the set of hyperedges adjacent to $e_i$ by $\nei:=\{e_j\in E: e_i\cap e_j \neq \varnothing\}$ and the number of such hyperedges by $|\nei|$.
Similarly, we say three hyperedges $e_i$, $e_j$, and $e_k$ are {\it connected} if one of them is adjacent to two the others. 

\smallsection{\Hwedges:}
We define a \textit{\hwedge} as an unordered pair of adjacent hyperedges. We denote the set of \hwedges in $G$ by $\wedge:=\{\{e_i,e_j\}\in {E \choose 2}: e_i\cap e_j \neq \varnothing\}$.
We use $\wij\in \wedge$ to denote the \hwedge consisting of $e_i$ and $e_j$. 
In the example hypergraph in Figure~\ref{fig:example:hypergraph},
there are four \hwedges: $\wedge_{12}$, $\wedge_{13}$, $\wedge_{23}$, and $\wedge_{14}$.

\begin{figure}[t!]
	\centering
	\includegraphics[width=0.75\columnwidth]{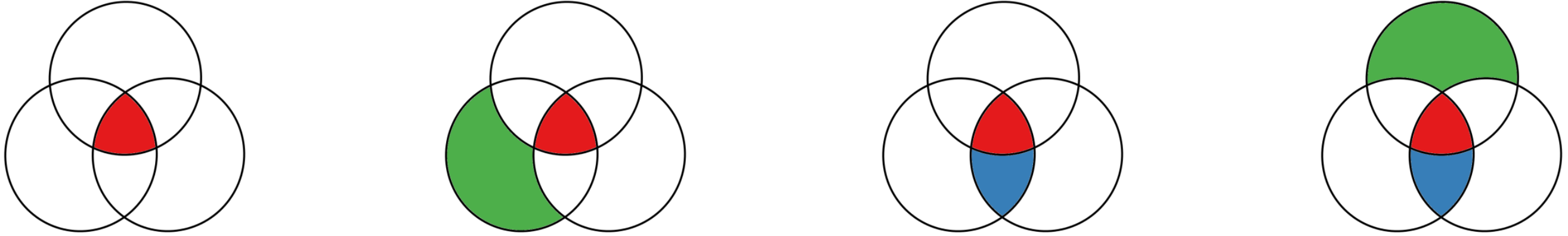} 
	\caption{\label{motif_three_hyperedges_duplicated}
		The \motifs whose instances contain duplicated hyperedges.
	}
\end{figure}



\smallsection{Projected Graph:}
We define the {\it projected graph} of $G=(V,E)$ as $\GT=(E,\wedge, \omega)$ where $\wedge$ is the set of hyperwedges and $\omega(\wij):=|e_i \cap e_j|$.
That is, in the projected graph $\GT$, hyperedges in $G$ act as nodes, and two of them are adjacent if and only if they share any member.
Note that for each hyperedge $e_i\in E$, $\nei$  is the set of neighbors of $e_i$ in $\GT$, and $|\nei|$ is its degree in $\GT$.
Figure~\ref{fig:example:graph} shows the projected graph of the example hypergraph in Figure~\ref{fig:example:hypergraph}.


\vspace{-0.5mm}
\subsection{Hypergraph Motifs}
\label{sec:concept:motif}
We introduce hypergraph motifs, which are basic building blocks of hypergraphs, with related concepts. Then, we discuss their properties and generalization.

\smallsection{Definition and Representation:}
Hypergraph motifs (or \motifs in short) are for describing the connectivity patterns of three connected hyperedges.
Specifically, given a set $\eijk$ of three connected hyperedges, \motifs describe its connectivity pattern by the emptiness of the following seven sets: (1) $e_i\setminus e_j \setminus e_k$, (2) $e_j\setminus e_k \setminus e_i$, (3) $e_k\setminus e_i \setminus e_j$, (4) $e_i\cap e_j \setminus e_k$, (5) $e_j\cap e_k \setminus e_i$, (6) $e_k\cap e_i \setminus e_j$, and (7) $e_i\cap e_j \cap e_k$.
Formally, a \motif is defined as a binary vector of size $7$ whose elements represent the emptiness of the above sets, respectively, and as seen in Figure~\ref{fig:example:motif}, \motifs are naturally represented in the Venn diagram.
While there can be $2^7$ \motifs, $26$ \motifs remain once we exclude symmetric ones, those with duplicated hyperedges (see Figure~\ref{motif_three_hyperedges_duplicated}), and those cannot be obtained from connected hyperedges.
The 26 cases, which we call {\it \motif 1} through {\it \motif 26}, are visualized in the Venn diagram in Figure~\ref{motif_three_hyperedges}.


\smallsection{Instances, Open \motifs, and Closed \motifs:}
Consider a hypergraph $G=(V,E)$.
A set of three connected hyperedges is an {\it instance} of \motif $t$ if their connectivity pattern corresponds to \motif $t$.
The count of each \motif's instances is used to characterize the local structure of $G$, as discussed in the following sections.
A \motif is {\it closed} if all three hyperedges in its instances are adjacent to (i.e., overlapped with) each other. If its instances contain two non-adjacent (i.e., disjoint) hyperedges, a \motif is {\it open}.
In Figure~\ref{motif_three_hyperedges}, \motifs $17$ - $22$ are open; the others are closed.

\smallsection{Properties of \motifs:}
From the definition of \motifs, the following desirable properties are immediate:
\vspace{-0.5mm}
\begin{itemize}[leftmargin=*]
\itemsep-0.2em 
\item {\bf Exhaustive:} \motifs capture connectivity patterns of \textit{all possible} three connected hyperedges.
\item {\bf Unique:} connectivity pattern of any three connected hyperedges is captured by \textit{at most one} \motif.
\item {\bf Size Independent:} \motifs capture connectivity patterns \textit{independently of the sizes of hyperedges}. Note that there can be infinitely many combinations of sizes of three connected hyperedges.
\end{itemize}
\vspace{-0.5mm}
Note that the exhaustiveness and the uniqueness imply that connectivity pattern of any three connected hyperedges is captured by \textit{exactly one} \motif.

\smallsection{Why Non-pairwise Relations?:}
Non-pairwise relations \\ (e.g., the emptiness of $e_{1}\cap e_{2} \cap e_{3}$ and $e_{1}\setminus e_{2} \setminus e_{3}$) play a key role in capturing the local structural patterns of real-world hypergraphs.
Taking only the pairwise relations (e.g., the emptiness of $e_{1}\cap e_{2}$,  $e_{1}\setminus e_{2}$, and $e_{2} \setminus  e_{1}$) into account limits the number of possible connectivity patterns of three distinct hyperedges to just eight,\footnote{\scriptsize \change{Note that using the conventional network motifs in projected graphs limits this number to two.}}  significantly limiting their expressiveness and thus usefulness.
Specifically, $12$ (out of $26$) \motifs have the same pairwise relations, while their occurrences and significances vary substantially in real-world hypergraphs. 
For example, in Figure~\ref{fig:example},  $\{e_{1}, e_{2}, e_{4}\}$  and $\{e_{1}, e_{3}, e_{4}\}$ have the same pairwise relations, while their connectivity patterns are distinguished by \motifs.

\smallsection{Generalization to More than $3$ Hyperedges:}
The concept of \motifs is easily generalized to four or more hyperedges. For example, a \motif for four hyperedges can be defined as a binary vector of size $15$ indicating the emptiness of each region in the Venn diagram for four sets.
After excluding disconnected ones, symmetric ones, and those with duplicated hyperedges, there remain $1,853$ and $18,656,322$ \motifs for four and five hyperedges, respectively, as discussed in detail in Appendix F of \cite{full}.
This work focuses on the \motifs for three hyperedges, which are already capable of characterizing local structures of real-world hypergraphs, as shown empirically in Section~\ref{sec:exp}.

\begin{figure*}[t]
	\vspace{-3mm}
\end{figure*}

\vspace{-0.5mm}
\subsection{Characteristic Profile (CP)}

What are the structural design principles of real-world hypergraphs distinguished from those of random hypergraphs?
Below, we introduce the characteristic profile (CP), which is a tool for answering the above question using \motifs.

\smallsection{Randomized Hypergraphs:}
While one might try to characterize the local structure of a hypergraph by absolute counts of each \motif's instances in it, some \motifs may naturally have many instances.
Thus, for more accurate characterization, we need random hypergraphs to be compared against real-world hypergraphs.
\change{We obtain such random hypergraphs by randomizing the compared real-world hypergraph. To this end, we represent the hypergraph $G=(V,E)$ as the bipartite graph $G'$ where $V$ and $E$ are the two subsets of nodes, and there exists an edge between $v\in V$ and $e \in E$ if and only if $v\in e$ That is, $G'=(V',E')$ where $V':=V\cup E$ and $E':=\{(v,e)\in V\times E: v\in e\}$.
Then, we use the Chung-Lu model to generate bipartite graphs where the degree distribution of $G'$ is well preserved \cite{aksoy2017measuring}. 
Reversely, from each of the generated bipartite graphs, we can obtain a randomized hypergraph where the degree (i.e., the number of hyperedges that each node belongs to) distribution of nodes and the size distribution of hyperedges in $G$ are well preserved. We provide the pseudocode and the distributions in the randomized hypergraphs in Appendix D of \cite{full}.}

\smallsection{Significance of \Motifs:}
We measure the significance of each \motif in a hypergraph by comparing the count of its instances against the count of them in randomized hypergraphs.
Specifically, the {\it significance} of a \motif $t$ in a hypergraph $G$ is defined as
	\vspace{-1mm}
\begin{equation}
    \Delta_t := \frac{M[t] - M_{rand}[t]}{M[t] + M_{rand}[t] + \epsilon}, \label{eq:significance}
    	\vspace{-1mm}
\end{equation}
where $M[t]$ is the number of instances of \motif $t$ in $G$, and $M_{rand}[t]$ is the average number of instances of \motif $t$ in randomized hypergraphs. We fixed $\epsilon$ to $1$ throughout this paper.
This way of measuring significance was proposed for network motifs as an alternative of normalized Z scores, which heavily depend on the graph size  \cite{milo2004superfamilies}.


\smallsection{Characteristic Profile (CP):} By normalizing and concatenating the significances of all \motifs in a hypergraph, we obtain the characteristic profile (CP), which summarizes the local structural pattern of the hypergraph.
Specifically, the {\it characteristic profile} of a hypergraph $G$ is a vector of size $26$, where each $t$-th element is 
	\vspace{-1mm}
\begin{equation}
    CP_t := \frac{\Delta_t}{\sqrt{\sum_{t=1}^{26} \Delta_t^2}}. \label{eq:cp}
    	\vspace{-1mm}
\end{equation}
Note that, for each $t$, $CP_t$ is between $-1$ and $1$.
The CP is used in Section~\ref{sec:exp:domain} to compare the local structural patterns of real-world hypergraphs from diverse domains.



\vspace{-0.5mm}
\section{Proposed Algorithms}
\label{sec:method}


Given a hypergraph, how can we count the instances of each \motif? 
Once we count them in the original and randomized hypergraphs, the significance of each \motif and the CP are obtained immediately by Eq.~\eqref{eq:significance} and Eq.~\eqref{eq:cp}.

In this section, we present \method (\textbf{Mo}tif \textbf{C}ounting in \textbf{Hy}pergraphs), which is a family of parallel algorithms for counting the instances of each \motif in the input hypergraph.
We first describe hypergraph projection, which is a preprocessing step of every version of \method.
Then, we present \methodE, which is for exact counting.
After that, we present two different versions of  \methodA, which are sampling-based algorithms for approximate counting.
Lastly, we discuss parallel and on-the-fly implementations.

Throughout this section, we use $\hijk$ to denote the \motif that describes the connectivity pattern of an \motif instance $\eijk$. We also use $\MT$ to denote the count of instances of \motif $t$.

\begin{algorithm}[t]
	\setstretch{1.25}
	\small
	\caption{\small Hypergraph Projection (Preprocess)}
	\label{hyperwedge_counting}
	\SetAlgoLined
	\SetKwInOut{Input}{Input}
	\SetKwInOut{Output}{Output}
	\nonl \hspace{-4mm} \Input{input hypergraph: $G=(V,E)$}
	\nonl \hspace{-4mm} \Output{projected graph: $\GT=(E,\PT, \omega)$}
	$\PT\leftarrow \varnothing$ \\
	$\omega \leftarrow$ map whose default value is $0$ \\
	\For{\textbf{each} hyperedge $e_i \in E$ \label{hyperwedge_counting:loop1}}{  
		\For{\textbf{each} node $v \in e_i$ \label{hyperwedge_counting:loop2}}{
			\For{\textbf{each} hyperedge $e_j \in E_v$ where $j>i$ \label{hyperwedge_counting:loop3}}{
				$\PT\leftarrow \PT \cup \{\wij\}$ \\ \label{hyperwedge_counting:body1}
				$\omega(\wij)=\omega(\wij)+ 1$
				\label{hyperwedge_counting:body2}
			}
		}
	}
	return $\GT=(E,\PT, \omega)$ 
\end{algorithm}

\smallsection{Remarks:}
	The problem of counting \motifs' occurrences bears some similarity to the classic problem of counting network motifs' occurrences.
	However, different from network motifs, which are defined solely based on pairwise interactions, \motifs are defined based on triple-wise interactions (e.g., $e_i\cap e_j \cap e_k$). 
	One might hypothesize that our problem can easily be reduced to the problem of counting the occurrences of network motifs, and thus existing solutions (e.g., \cite{bressan2019motivo,pinar2017escape}) are applicable to our problem.	
	In order to examine this possibility, we consider the following two attempts:
	\vspace{-0.5mm}
	\begin{enumerate}[leftmargin=*]
		\itemsep-0.2em 
		{\setlength\itemindent{5pt}\item[(a)] Represent pairwise relations between hyperedges using the projected graph, where each edge $\{e_i,e_j\}$ indicates $e_i\cap e_j\neq \emptyset$. }
		{\setlength\itemindent{5pt}\item[(b)] Represent pairwise relations between hyperedges using the directed projected graph where each directed edge $e_i \rightarrow e_j$ indicates 
		$e_i\cap e_j\neq \emptyset$ and at the same time $e_i \not\subset e_j$.}
	\end{enumerate}
	\vspace{-0.5mm}
	The number of possible connectivity patterns (i.e., network motifs) among three distinct connected hyperedges is just two (i.e., closed and open triangles) and eight in (a) and (b), respectively.
	In both cases,
	instances of multiple \motifs are not distinguished by network motifs, and 
	the occurrences of \motifs can not be inferred from those of network motifs.

In addition, another computational challenge stems from the fact that triple-wise and even pair-wise relations between hyperedges need to be computed from the input hypergraph, while pairwise relations between edges are given in graphs. This challenge necessitates the precomputation of partial relations, described in the next subsection.
	


\begin{figure*}[t]
	\vspace{-3mm}
\end{figure*}

\vspace{-0.5mm}
\subsection{Hypergraph Projection (Algorithm~\ref{hyperwedge_counting})}
\label{sec:method:projection}
As a preprocessing step, every version of \method builds the projected graph $\GT=(E,\PT, \omega)$ (see Section~\ref{sec:concept:prelim}) of the input hypergraph $G=(V,E)$, as described in Algorithm~\ref{hyperwedge_counting}. 
To find the neighbors of each hyperedge $e_i$ (line~\ref{hyperwedge_counting:loop1}), the algorithm visits each hyperedge $e_j$ that contains $v$ and satisfies $j>i$ (line~\ref{hyperwedge_counting:loop3}) for each node $v\in e_i$ (line~\ref{hyperwedge_counting:loop2}).
Then for each such $e_j$, it adds $\wij=\{e_i,e_j\}$ to $\PT$ and increments $\omega(\wij)$ (lines~\ref{hyperwedge_counting:body1} and \ref{hyperwedge_counting:body2}).
The time complexity of this preprocessing step is given in Lemma~\ref{lemma:projection:time}.
\begin{lemma}[Complexity of Hypergraph Projection\label{lemma:projection:time}]
	The time complexity of Algorithm~\ref{hyperwedge_counting} is $O(\sum_{\wij\in\PT}|e_i\cap e_j|)$. \vspace{-2mm}
\end{lemma}
\begin{proof}
If all sets and maps are implemented using hash tables, lines~\ref{hyperwedge_counting:body1} and \ref{hyperwedge_counting:body2} take $O(1)$ time, and they are executed $|e_i\cap e_j|$ times for each $\wij\in \PT$. \vspace{-2mm}
\end{proof}
\noindent Since $|\PT|<\sum_{e_i\in E}\degt{e_i}$ and $|e_i\cap e_j|\leq|e_i|$, Eq.~\eqref{eq:projection:time} holds.
\begin{equation}
\sum\nolimits_{\wij\in\PT}|e_i\cap e_j| < \sum\nolimits_{e_i\in E}(|e_i|\cdot|\nei|). \label{eq:projection:time} 
	\vspace{-1mm}
\end{equation}


\begin{algorithm}[t]
	\setstretch{1.25}
	\small
	\caption{\small \methodE: Exact \Motif Counting \label{exact_motif_counting}}
	\SetAlgoLined
	\SetKwInOut{Input}{Input}
	\SetKwInOut{Output}{Output}
	\nonl \hspace{-4mm} \Input{ \ \ (1) input hypergraph: $G=(V,E)$ \\ (2) projected graph: $\GT=(E,\PT, \omega)$}
	\nonl \hspace{-4mm} \Output{exact count of each \motif $t$'s instances: $\MT$}
	$M \leftarrow$ map whose default value is $0$ \\
	\For{\textbf{each} hyperedge $e_i \in E$ \label{exact_motif:loop1}}{
		\For{\textbf{each} unordered hyperedge pair $\{e_j, e_k\} \in$ $\nei \choose 2$\label{exact_motif:loop2}}{
			\If{$e_j \cap e_k = \varnothing$ or $i<\min(j,k)$\label{exact_motif:condition}}{
				$M[\hijk] \mathrel{+}= 1$ \label{exact_motif:count}
			}
		}
	}
	\textbf{return} $M$ 
\end{algorithm}

\vspace{-0.5mm}
\subsection{Exact \Motif Counting (Algorithm~\ref{exact_motif_counting})}
\label{sec:method:exact}

We present \methodE (\method \textbf{E}xact), which counts the instances of each \motif exactly. 
The procedures of \methodE are described in Algorithm~\ref{exact_motif_counting}.
For each hyperedge $e_i\in E$ (line~\ref{exact_motif:loop1}), each unordered pair $\{e_j,e_k\}$ of its neighbors, where $\eijk$ is an \motif instance, is considered (line~\ref{exact_motif:loop2}).
If $e_j \cap e_k = \varnothing$ (i.e., if the corresponding \motif is open), $\eijk$ is considered only once.
However, if $e_j \cap e_k \neq \varnothing$ (i.e., if the corresponding \motif is closed),
$\eijk$ is considered two more times (i.e., when $e_j$ is chosen in line~\ref{exact_motif:loop1} and when $e_k$ is chosen in line~\ref{exact_motif:loop1}).
Based on these observations, given an \motif instance $\eijk$, the corresponding count $M[\hijk]$ is incremented (line~\ref{exact_motif:count}) only if $e_j \cap e_k = \varnothing$ or $i < \min(j,k)$ (line~\ref{exact_motif:condition}). This guarantees that each instance is counted exactly once. 
The time complexity of \methodE is given in Theorem~\ref{thm:exact:time}, which uses Lemma~\ref{lemma:motif:time}.

\begin{lemma}[Time Complexity of Computing $\hijk$]\label{lemma:motif:time}
	Given the input hypergraph $G=(V,E)$ and its projected graph $\GT=(E,\wedge,\omega)$, for each \motif instance $\eijk$, computing $\hijk$ takes $O(\min(|e_i|,|e_j|,|e_k|))$ time. \vspace{-2mm}
\end{lemma}
\begin{proof}
	Assume $|e_i|=\min(|e_i|,|e_j|,|e_k|)$, without loss of generality, and all sets and maps are implemented using hash tables.
	As defined in Section~\ref{sec:concept:motif}, $\hijk$ is computed in $O(1)$ time from the emptiness of the following sets:
	(1) $e_i\setminus e_j \setminus e_k$, (2) $e_j\setminus e_k \setminus e_i$, (3) $e_k\setminus e_i \setminus e_j$, (4) $e_i\cap e_j \setminus e_k$, (5) $e_j\cap e_k \setminus e_i$, (6) $e_k\cap e_i \setminus e_j$, and (7) $e_i\cap e_j \cap e_k$.
	We check their emptiness from their cardinalities.
	We obtain $e_i$, $e_j$, and $e_k$, which are stored in $G$, and their cardinalities in $O(1)$ time.
	Similarly, we obtain $|e_i \cap e_j|$, $|e_j \cap e_k|$, and $|e_k \cap e_i|$, which are stored in $\GT$, in $O(1)$ time.
	Then, we compute $|e_i \cap e_j \cap e_k|$ in $O(|e_i|)$ time by checking for each node in $e_i$ whether it is also in both $e_j$ and $e_k$.
	From these cardinalities, we obtain the cardinalities of the six other sets in $O(1)$ time as follows: 
	\begin{align*}
		& (1) \ |e_i\setminus e_j \setminus e_k|  = |e_i|-|e_i\cap e_j|-|e_k\cap e_i|+|e_i \cap e_j \cap e_k|, \\
		& (2) \ |e_j\setminus e_k \setminus e_i| = |e_j|-|e_i\cap e_j|-|e_j\cap e_k|+|e_i \cap e_j \cap e_k|, \\
		& (3) \ |e_k\setminus e_i \setminus e_j| = |e_k|-|e_k\cap e_i|-|e_j\cap e_k|+|e_i \cap e_j \cap e_k|, \\
		& (4) \ |e_i\cap e_j \setminus e_k| = |e_i\cap e_j| - |e_i \cap e_j \cap e_k|, \\
		& (5) \ |e_j\cap e_k \setminus e_i| = |e_j\cap e_k| - |e_i \cap e_j \cap e_k|, \\
		&(6) \ |e_k\cap e_i \setminus e_j| = |e_k\cap e_i| - |e_i \cap e_j \cap e_k|.
	\end{align*}
	Hence, the time complexity of computing $\hijk$ is $O(|e_i|)=O(\min(|e_i|,|e_j|,|e_k|))$. \vspace{-2mm}
\end{proof}

\begin{thm}[Complexity of \methodE]\label{thm:exact:time}
	The time complexity of Algorithm~\ref{exact_motif_counting} is $O(\sum_{e_i \in E}(|\nei|^2 \cdot |e_i|))$. \vspace{-2mm}
\end{thm}
\begin{proof}
 Assume all sets and maps are implemented using hash tables.
 The total number of triples $\eijk$ considered in line~\ref{exact_motif:loop2} is $O(\sum_{e_i \in E}|\nei|^2)$.
 By Lemma~\ref{lemma:motif:time}, for such a triple $\eijk$, computing $\hijk$ takes $O(|e_i|)$ time.
 Thus, the total time complexity of Algorithm~\ref{exact_motif_counting} is $O(\sum_{e_i \in E}(|e_i|\cdot|\nei|^2))$, which dominates that of the preprocessing step (see Lemma~\ref{lemma:projection:time} and Eq.~\eqref{eq:projection:time}). \vspace{-2.5mm}
\end{proof}	
\smallsection{Extension of \methodE to \Motif Enumeration:} \\ \change{Since \methodE visits all \motif instances to count them, it is extended to the problem of enumerating every \motif instance (with its corresponding \motif), as described in Algorithm~\ref{alg:enum}.
The time complexity remains the same.}

\begin{algorithm}[t]
	\setstretch{1.25}
	\small
	\caption{\label{alg:enum}\small \methodEN for \Motif Enumeration}
	\SetAlgoLined
	\SetKwInOut{Input}{Input}
	\SetKwInOut{Output}{Output}
		\nonl \hspace{-4mm} \Input{ \ \ (1) input hypergraph: $G=(V,E)$ \\ (2) projected graph: $\GT=(E,\PT, \omega)$}
	\nonl \hspace{-4mm} \Output{\motif instances and their corresponding \motifs}
	\For{\textbf{each} hyperedge $e_i \in E$ \label{enum_alg:loop1}}{
		\For{\textbf{each} unordered hyperedge pair $\{e_j, e_k\} \in$ $\nei \choose 2$\label{enum_alg:loop2}}{
			\If{$e_j \cap e_k = \varnothing$ or $i<\min(j,k)$\label{enum_alg:condition}}{
				write($e_i$, $e_j$, $e_k$, $\hijk$)\label{enum_alg:write}
			}
		}
	}
\end{algorithm}

\IncMargin{0.5em}
\begin{algorithm}[t]	
	\setstretch{1.25}
	\small
	\caption{\small \methodAE: Approximate \Motif Counting Based on Hyperedge Sampling}
	\label{sampling_ver1}
	\SetAlgoLined
	\SetKwInOut{Input}{Input}
	\SetKwInOut{Output}{Output}
	\nonl \hspace{-5mm} \Input{ \ \ (1) input hypergraph: $G=(V,E)$ \\ (2) projected graph: $\GT=(E,\PT, \omega)$ \\  (3) number of samples: $s$}
	\nonl \hspace{-5mm} \Output{estimated count of each \motif $t$'s instances: $\MBT$}
	$\MBT\leftarrow$ map whose default value is $0$\\ 
	\For{$n\leftarrow 1...s$}{
		$e_i\leftarrow$ sample a uniformly random hyperedge \label{sampling_ver1:sample} \\
		\For{\textbf{each} hyperedge $e_j \in N_{e_i}$}{\label{sampling_ver1:loop1:start}
			\For{\textbf{each} hyperedge $e_k \in (N_{e_i} \cup N_{e_j} \setminus \{e_i,e_j\})$}{
				\If{$e_k \not\in N_{e_i}$ or $j<k$\label{sampling_ver1:condition}}{
					$\MB[\hijk] \mathrel{+}= 1$ 
					\label{sampling_ver1:count}
				}
			}
		}
	}
	\label{sampling_ver1:loop1:end}
	\For{\textbf{each} \motif $t$}{\label{sampling_ver1:scale:start}
		$\MBT \leftarrow \MBT \cdot \tfrac{|E|}{3s}$
	}
	\label{sampling_ver1:scale:end}
	\textbf{return} $\MB$ 
\end{algorithm}
\DecMargin{0.5em}

\vspace{-0.5mm}
\subsection{Approximate \Motif Counting}
\label{sec:method:approx}

We present two different versions of \methodA (\method \textbf{A}pproximate), which approximately count the instances of each \motif.
Both versions yield unbiased estimates of the counts by exploring the input hypergraph partially through hyperedge and \hwedge sampling, respectively.



\begin{figure*}[t]
	\vspace{-3mm}
\end{figure*}

\smallsection{\methodAE: Hyperedge Sampling (Algorithm~\ref{sampling_ver1}):}\label{sec:method:approx:ver1} 

\noindent \methodAE (Algorithm~\ref{sampling_ver1}) is based on hyperedge sampling. 
It repeatedly samples $s$ hyperedges from the hyperedge set $E$ uniformly at random with replacement (line~\ref{sampling_ver1:sample}). 
For each sampled hyperedge $e_i$, the algorithm searches for all \motif instances that contain $e_i$ (lines~\ref{sampling_ver1:loop1:start}-\ref{sampling_ver1:loop1:end}), and to this end, the $1$-hop and $2$-hop neighbors of $e_i$ in the projected graph $\GT$ are explored. After that, for each such instance $\eijk$ of h-motif $t$, the corresponding count $\MBT$ is incremented (line~\ref{sampling_ver1:count}). 
Lastly, each estimate $\MBT$ is rescaled by multiplying it with $\frac{|E|}{3s}$ (lines~\ref{sampling_ver1:scale:start}-\ref{sampling_ver1:scale:end}), which is the reciprocal of the expected number of times that each of the \motif $t$'s instances is counted.\footnote{\scriptsize Each hyperedge is expected to be sampled $\frac{s}{|E|}$ times, and each \motif instance is counted whenever any of its $3$ hyperedges is sampled.}
This rescaling makes each estimate $\MBT$ unbiased, as formalized in Theorem~\ref{thm:sampling_ver1:accuracy}.

\vspace{-1mm}
\begin{thm}[Bias and Variance of \methodAE]
	\label{thm:sampling_ver1:accuracy}
	For every \motif t, Algorithm~\ref{sampling_ver1} provides an unbiased estimate $\MBT$ of the count $\MT$ of its instances, i.e.,	
	\vspace{-1mm}
	\begin{equation}
	\mathbb{E}[\MBT]=\MT. \label{sampling_ver1:bias}
		\vspace{-1mm}
	\end{equation}
	The variance of the estimate is
	\vspace{-1mm}
	\begin{equation}
	\mathbb{V}\mathrm{ar}[\MBT] = \frac{1}{3s}\cdot \MT\cdot (|E|-3)+\frac{1}{9s}\sum_{l=0}^{2}p_l[t] \cdot(l|E|-9),
	\label{sampling_ver1:variance}
	\vspace{-1mm}
	\end{equation}
	where $p_l[t]$ is the number of pairs of \motif $t$'s instances that share $l$ hyperedges. \vspace{-2mm}
\end{thm}
\begin{proof}
	See Appendix~\ref{sampling_ver1:proof}. \vspace{-2mm}
\end{proof}
\noindent The time complexity of \methodAE is given in Theorem~\ref{thm:sampling_ver1:time}.
\begin{thm}[Complexity of \methodAE] \label{thm:sampling_ver1:time}
	The average time complexity of Algorithm~\ref{sampling_ver1} is $O(\frac{s}{|E|}\sum_{e_i\in E}(|e_i|\cdot|\nei|^2))$. \vspace{-2mm}
\end{thm}
\begin{proof}
	Assume all sets and maps are implemented using hash tables.
	For a sample hyperedge $e_i$, computing $\nei \cup \nej$ for every $e_j \in \nei$ takes $O(\sum_{e_j\in \nei} (|\nei\cup\nej|))$ time, and by Lemma~\ref{lemma:motif:time}, computing $\hijk$ for all considered \motif instances takes $O(\min(|e_i|, |e_j|)\cdot \sum_{e_j\in \nei} |\nei\cup\nej|)$ time.
	Thus, from $|\nei\cup\nej| \leq |\nei|+|\nej|$, the time complexity for processing a sample $e_i$ is 
	\vspace{-1mm}
	\begin{multline*}
	O(\min(|e_i|, |e_j|)\cdot \sum\nolimits_{e_j\in \nei} (|\nei|+|\nej|)) \\ =O(|e_i|\cdot|\nei|^2 + \sum\nolimits_{e_j\in \nei}(|e_j|\cdot|\nej|)),
		\vspace{-1mm}
	\end{multline*} 
	which can be written as
	\vspace{-1mm}
	\begin{multline*}
	O(\sum\nolimits_{e_i\in E}(\mathbb{1}(e_i \text{ is sampled})\cdot|e_i|\cdot|\nei|^2) \\ + \sum\nolimits_{e_j\in E}(\mathbb{1}(e_j \text{ is adjacent to the sample})\cdot|e_j|\cdot|\nej|)).
		\vspace{-1mm}
	\end{multline*}
	From this, linearity of expectation, $\mathbb{E}[\mathbb{1}(e_i$ is sampled$)]=\frac{1}{|E|}$, and $\mathbb{E}[\mathbb{1}(e_j$ is adjacent to the sample$)]=\frac{|\nej|}{|E|}$, the average time complexity per sample hyperedge becomes
	$O(\frac{1}{|E|}$ $\sum_{e_i\in E}(|e_i|\cdot|\nei|^2))$.
	Hence, the total time complexity for processing $s$ samples is $O(\frac{s}{|E|}\sum_{e_i\in E}(|e_i|\cdot|\nei|^2))$.\qedhere	\vspace{-2mm}
\end{proof}

\IncMargin{0.5em}
\begin{algorithm}[t]
	\setstretch{1.25}
	\small
	\caption{\small \methodAW: Approximate \Motif Counting Based on \Hwedge Sampling}
	\label{sampling_ver2}
	\SetAlgoLined
	\SetKwInOut{Input}{Input}
	\SetKwInOut{Output}{Output}
	\nonl \hspace{-5mm} \Input{ \ \ (1) input hypergraph: $G=(V,E)$ \\ (2) projected graph: $\GT=(E,\PT, \omega)$ \\  (3) number of samples: $r$}
	\nonl \hspace{-5mm} \Output{estimated count of each \motif $t$'s instances: $\MHT$}
	$\MH\leftarrow$ map whose default value is $0$ \\
	\For{$n\leftarrow 1...r$\label{sampling_ver2:loop1}}{
		$\wedge_{ij}\leftarrow$ a uniformly random \hwedge \label{sampling_ver2:sample} \\ 
		\For{\textbf{each} hyperedge $e_k \in (N_{e_i} \cup N_{e_j} \setminus \{e_i,e_j\})$\label{sampling_ver2:loop2}}{
			$\MH[\hijk] \mathrel{+}= 1$
			\label{sampling_ver2:count}
		}
		\label{sampling_ver2:loop2:end}
	}
	\For{\textbf{each} \motif $t$\label{sampling_ver2:rescale:start}}{
		\eIf(\hfill $\triangleright$ \texttt{\color{blue}open \motifs}){17 $\leq$ t $\leq$ 22}{ 
			$\MHT \leftarrow \MHT \cdot \tfrac{|\wedge|}{2r}$ 
		}
		(\hfill $\triangleright$ \texttt{\color{blue}closed \motifs}){
			$\MHT \leftarrow \MHT \cdot \tfrac{|\wedge|}{3r}$
		}
	}\label{sampling_ver2:rescale:end}
	\textbf{return} $\MH$ 
\end{algorithm}
\DecMargin{0.5em}

\smallsection{\methodAW: \Hwedge Sampling (Algorithm~\ref{sampling_ver2}):}

\noindent \methodAW (Algorithm~\ref{sampling_ver2})  provides a better trade-off between speed and accuracy than \methodAE.
Different from \methodAE, which samples hyperedges, \methodAW is based on \hwedge sampling. 
It selects $r$ \hwedges uniformly at random with replacement (line~\ref{sampling_ver2:sample}), and for each sampled \hwedge $\wedge_{ij} \in \wedge$, it searches for all \motif instances that contain $\wedge_{ij}$ (lines~\ref{sampling_ver2:loop2}-\ref{sampling_ver2:loop2:end}). To this end, the hyperedges that are adjacent to $e_i$ or $e_j$ in the projected graph $\GT$ are considered (line~\ref{sampling_ver2:loop2}). For each such instance $\eijk$ of \motif $t$, the corresponding estimate $\MHT$ is incremented (line~\ref{sampling_ver2:count}). 
Lastly, each estimate $\MHT$ is rescaled so that it unbiasedly estimates $\MT$, as formalized in Theorem~\ref{thm:sampling_ver2:accuracy}.
To this end, it is multiplied by the reciprocal of the expected number of times that each instance of \motif $t$ is counted.\footnote{\scriptsize Note that each instance of open and closed \motifs contains $2$ and $3$ \hwedges, respectively.
	Each instance of closed \motifs is counted if one of the $3$ \hwedges in it is sampled, while that of open \motifs is counted if one of the $2$ \hwedges in it is sampled.
	Thus, on average, each instance of open and closed \motifs is counted ${3r}/{|\wedge|}$ and ${2r}/{|\wedge|}$ times, respectively.}

\vspace{-1mm}
\begin{thm}[Bias and Variance of \methodAW]
	\label{thm:sampling_ver2:accuracy}
	For every \motif t, Algorithm~\ref{sampling_ver2} provides an unbiased estimate $\MHT$ of the count $\MT$ of its instances, i.e.,	
	\begin{equation}
	\mathbb{E}[\MHT]=\MT. \label{sampling_ver2:bias}
	\vspace{-1mm}
	\end{equation}
	For every closed \motif $t$, the variance of the estimate is
	\vspace{-1mm}
	\begin{equation}
	\mathbb{V}\mathrm{ar}[\MHT] = \frac{1}{3r}\cdot\MT\cdot(|\wedge|-3)+\frac{1}{9r}\sum_{n=0}^{1}q_n[t]\cdot(n|\wedge|-9), \label{sampling_ver2:variance:closed}
	\vspace{-1mm}
	\end{equation}
	where $q_n[t]$ is the number of pairs of \motif $t$'s instances that share $n$ hyperwedges.
	For every open \motif $t$, the variance is
	\vspace{-1mm}
	\begin{equation}
	\mathbb{V}\mathrm{ar}[\MHT] = \frac{1}{2r} \cdot \MT\cdot(|\wedge|-2)+\frac{1}{4r}\sum_{n=0}^{1}q_n[t]\cdot(n|\wedge|-4). \label{sampling_ver2:variance:open} 
	\vspace{-2mm}
	\end{equation}
\end{thm}
\begin{proof}
	See Appendix~\ref{sampling_ver2:proof}. \vspace{-2mm}
\end{proof}
\begin{figure*}[t]
	\vspace{-3mm}
\end{figure*}

\noindent The time complexity of \methodAW is given in Theorem~\ref{thm:sampling_ver2:time}.
\vspace{-1mm}
\begin{thm}[Complexity of \methodAW] \label{thm:sampling_ver2:time}
	The average time complexity of Algorithm~\ref{sampling_ver2} is $O(\frac{r}{|\wedge|}\sum_{e_i\in E}(|e_i|\cdot|\nei|^2))$. \vspace{-2mm}
\end{thm}
\begin{proof}
	Assume all sets and maps are implemented using hash tables.
	For a sample \hwedge $\wij$, computing $\nei \cup \nej$ takes $O(|\nei\cup\nej|)$ time, and by Lemma~\ref{lemma:motif:time}, computing $\hijk$ for all considered \motif instances takes $O(\min(|e_i|, |e_j|)\cdot |\nei\cup\nej|)$ time.
	Thus, from $|\nei\cup\nej| \leq |\nei|+|\nej|$, the time complexity for processing a sample $\wij$ is 
	$O(\min(|e_i|, |e_j|)\cdot (|\nei|+|\nej|))=O(|e_i|\cdot|\nei| + |e_j|\cdot|\nej|),$
	which can be written as
	\vspace{-1mm}
	\begin{multline*}
	O(\sum\nolimits_{e_i\in E}(\mathbb{1}(e_i \text{ is included in the sample})\cdot|e_i|\cdot|\nei|) \\ + \sum\nolimits_{e_j\in E}(\mathbb{1}(e_j \text{ is included in the sample})\cdot|e_j|\cdot|\nej|)).
		\vspace{-1mm}
	\end{multline*}
	From this, linearity of expectation, $\mathbb{E}[\mathbb{1}(e_i$ is included in the sample$)]=\frac{|\nei|}{|\wedge|}$, and $\mathbb{E}[\mathbb{1}(e_j$ is included in the sample$)]=\frac{|\nej|}{|\wedge|}$, the average time complexity per sample \hwedge is
	$O(\frac{1}{|\wedge|}\sum_{e_i\in E}(|e_i|\cdot|\nei|^2))$.
	Hence, the total time complexity for processing $r$ samples is $O(\frac{r}{|\wedge|}\sum_{e_i\in E}(|e_i|\cdot|\nei|^2))$.\qedhere	 \vspace{-2mm}
\end{proof}

\smallsection{Comparison of \methodAE and \methodAW:} \label{compare:sampling}
Empirically, \methodAW provides a better trade-off between speed and accuracy than \methodAE, as presented in Section~\ref{sec:exp:algo}. We provide an analysis that supports this observation. 
Assume that the numbers of samples in both algorithms are set so that $\alpha=\frac{s}{|E|}=\frac{r}{|\wedge|}$.
For each \motif $t$,
since both estimates $\MBT$ of \methodAE and $\MHT$ of \methodAW are unbiased (see Eqs.~\eqref{sampling_ver1:bias} and \eqref{sampling_ver2:bias}), we only need to compare their variances.
By Eq.~\eqref{sampling_ver1:variance}, $\mathbb{V}\mathrm{ar}[\MBT]=O(\frac{\MT+p_1[t]+p_2[t]}{\alpha})$, and by Eq.~\eqref{sampling_ver2:variance:closed} and Eq.~\eqref{sampling_ver2:variance:open}, $\mathbb{V}\mathrm{ar}[\MHT]=O(\frac{\MT+q_1[t]}{\alpha})$.
By definition, $q_1[t]\leq p_2[t]$, and thus $\frac{\MT+q_1[t]}{\alpha} \leq \frac{\MT+p_1[t]+p_2[t]}{\alpha}$. 
Moreover, in real-world hypergraphs, $p_1[t]$ tends to be several orders of magnitude larger than the other terms (i.e., $p_2[t]$, $q_1[t]$, and $\MT$), and thus $\MBT$ of \methodAE tends to have  larger variance (and thus larger estimation error) than $\MHT$ of \methodAW.
Despite this fact, as shown in Theorems~\ref{thm:sampling_ver1:time} and \ref{thm:sampling_ver2:time}, \methodAE and \methodAW have the same time complexity, $O(\alpha\cdot \sum_{e_i\in E}(|e_i|\cdot|\nei|^2))$.
Hence, \methodAW is expected to give a better trade-off between speed and accuracy than \methodAE, as confirmed empirically  in Section~\ref{sec:exp:algo}.

\begin{table}[t]
	\begin{center}
		\caption{\label{dataset_table}Statistics of 11 real hypergraphs from 5 domains.} 
		\vspace{-1mm}
		\scalebox{0.77}{
			\begin{tabular}{l|c|c|c|c|c}
				\toprule
				\textbf{Dataset} & \textbf{$|V|$} & \textbf{$|E|$} & \textbf{$|\bar{e}|$*} & \textbf{$|\wedge|$} & \textbf{\# \Motifs}\\
				\midrule
				coauth-DBLP & 1,924,991 & 2,466,792 & \change{25} & 125M & 26.3B $\pm$ 18M\\
				coauth-geology& 1,256,385 & 1,203,895 & \change{25} & 37.6M & 6B $\pm$ 4.8M\\
				coauth-history& 1,014,734 & 895,439 & \change{25} & 1.7M & 83.2M\\
				\midrule
				contact-primary  & 242 & 12,704 & \change{5} & 2.2M & 617M\\
				contact-high & 327 & 7,818 & \change{5} & 593K & 69.7M\\
				\midrule
				email-Enron & 143 & 1,512 & \change{18} & 87.8K & 9.6M\\
				email-EU & 998 & 25,027 & \change{25} & 8.3M & 7B\\
				\midrule
				tags-ubuntu & 3,029 & 147,222 & \change{5} & 564M & 4.3T $\pm$ 1.5B\\
				tags-math & 1,629 & 170,476 & \change{5} & 913M & 9.2T $\pm$ 3.2B\\
				\midrule
				threads-ubuntu & 125,602 & 166,999 & \change{14} & 21.6M & 11.4B\\
				threads-math & 176,445 & 595,749 & \change{21} & 647M & 2.2T $\pm$ 883M\\
				\bottomrule
				\multicolumn{6}{l}{$*$ The maximum size of a hyperedge.}
		\end{tabular}}
	\end{center}
	\vspace{-3mm}
\end{table}

\vspace{-0.5mm}
\subsection{Parallel and On-the-fly Implementations}\label{sec:method:par_fly}

We discuss parallelization of \method and then on-the-fly computation of projected graphs.

\smallsection{Parallelization:}
All versions of \method and hypergraph projection are easily parallelized.
Specifically, we can parallelize hypergraph projection and \methodE by letting multiple threads process different hyperedges (in line~\ref{hyperwedge_counting:loop1} of Algorithm~\ref{hyperwedge_counting} and line~\ref{exact_motif:loop1} of Algorithm~\ref{exact_motif_counting}) independently in parallel. 
Similarly, we can parallelize \methodAE and \methodAW by letting multiple threads sample and process different hyperedges (in line~\ref{sampling_ver1:sample} of Algorithm~\ref{sampling_ver1}) and \hwedges (in line~\ref{sampling_ver2:sample} of Algorithm~\ref{sampling_ver2}) independently in parallel. 
The estimated counts of the same \motif obtained by different threads are summed up only once before they are returned as outputs.
We present some empirical results in Section~\ref{sec:exp:algo}.

\smallsection{\Motif Counting without Projected Graphs:} \label{method:on_the_fly}
If the input hypergraph $G$ is large, computing its projected graph $\GT$ (Algorithm~\ref{hyperwedge_counting}) is time and space consuming.
Specifically, building $\GT$ takes $O(\sum_{\wij\in\PT}|e_i\cap e_j|)$ time (see Lemma~\ref{lemma:projection:time}) and requires $O(|E|+|\wedge|)$ space, which often exceeds $O(\sum_{e_i\in E}$ $|e_i|)$ space required for storing $G$.
Thus, instead of precomputing $\GT$ entirely, we can build it incrementally while memoizing partial results within a given memory budget.
For example, in \methodAW (Algorithm~\ref{sampling_ver2}), we compute the neighborhood of a hyperedge $e_{i}\in E$ in $\GT$ (i.e., $\{(k,\omega(\wik)):k\in\nei\}$) only if (1) a \hwedge with $e_{i}$ (e.g., $\wij$) is sampled (in line~\ref{sampling_ver2:sample}) and (2) its neighborhood is not memoized, \change{as described in the pseudocode in Appendix G of \cite{full}.
Whether they are computed on the fly or read from memoized results, we always use exact neighborhoods, and thus this change does not affect the accuracy of the considered algorithm.} 

This incremental computation of $\GT$ can be beneficial in terms of speed since it skips projecting the neighborhood of a hyperedge if no \hwedge containing it is sampled.
However, it can also be harmful if memoized results exceed the memory budget and some parts of $\GT$ need to be rebuilt multiple times.
Then, given a memory budget in bits, how should we prioritize hyperedges if all their neighborhoods cannot be memoized? 
According to our experiments, despite their large size, memoizing the neighborhoods of hyperedges with high degree in $\GT$ makes \methodAW faster than memoizing the neighborhoods of randomly chosen hyperedges or least recently used (LRU) hyperedges. 
We experimentally examine the effects of the memory budget in Section~\ref{sec:exp:algo}. 

\vspace{-0.5mm}
\section{Experiments}
\label{sec:exp}

\begin{table*}
	\centering
	\vspace{-4mm}
	\caption{\label{absolute_table} 
		Real-world and random hypergraphs have distinct distributions of \motif instances. We report the absolute counts of each \motif's instances in a real-world hypergraph from each domain and its corresponding random hypergraph. 
		To compare the counts in both hypergraphs, we measure the relative count (RC) of each \motif. 
		We also rank the counts, and we report each \motif's rank difference (RD) in the real-world and corresponding random hypergraphs.}
	\includegraphics[width=1.0\textwidth]{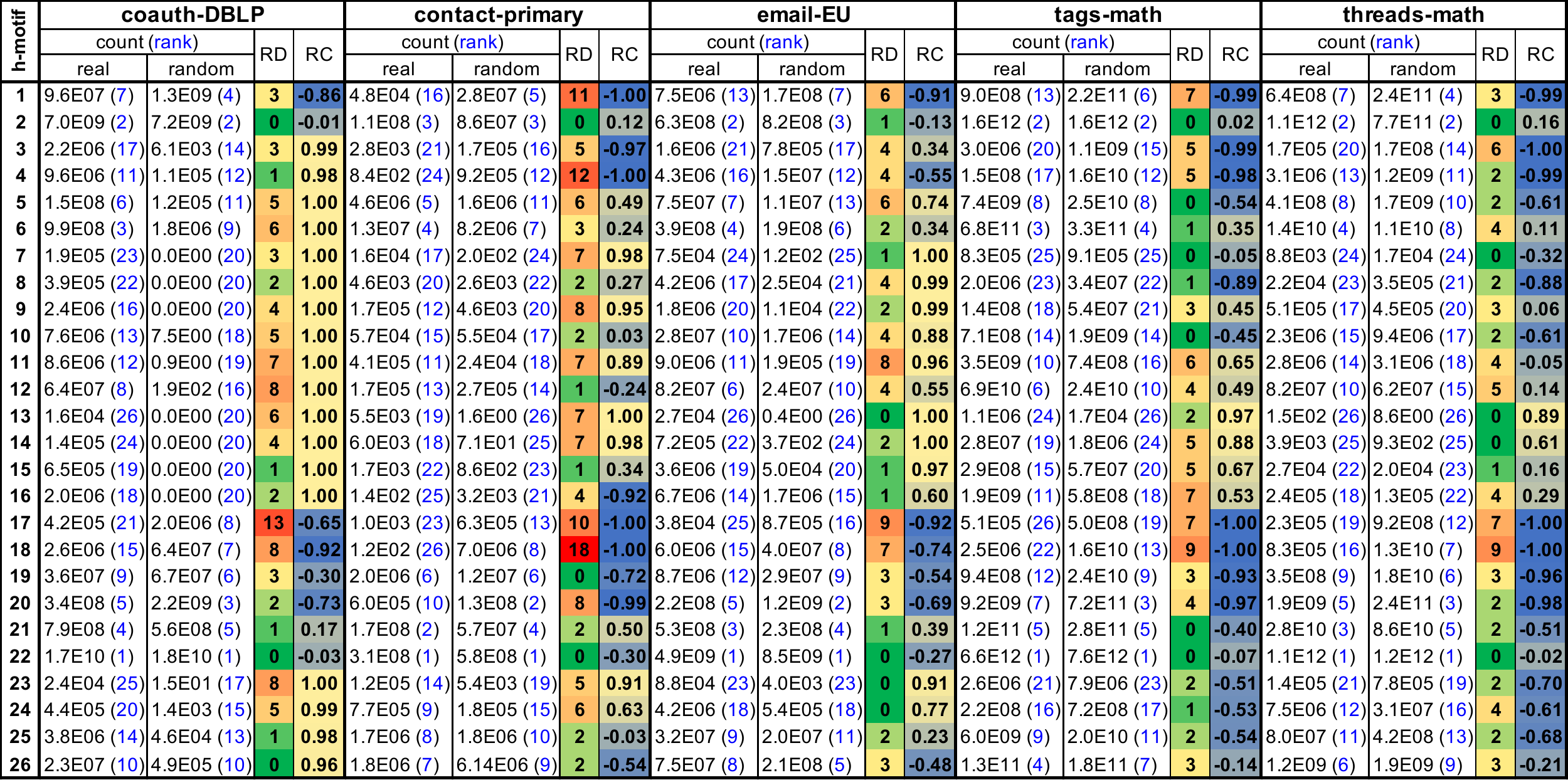}
\end{table*}

%

In this section, we review our experiments that we design for answering the following questions:
\vspace{-0.5mm}
\begin{itemize}[leftmargin=*]
	\itemsep-0.2em 
	\item {\bf Q1. Comparison with Random:} Does counting instances of different \motifs reveal structural design principles of real-world hypergraphs distinguished from those of random hypergraphs? 
	\item {\bf Q2. Comparison across Domains:} Do characteristic profiles capture local structural patterns of hypergraphs unique to each domain?
	\change{\item {\bf Q3. Observations and Applications:} What are interesting discoveries and applications of \motifs in real-world hypergraphs?}
	\item {\bf Q4. Performance of Counting Algorithms:} How fast and accurate are the different versions of \method? Does the advanced version outperform the basic ones?
\end{itemize}


\vspace{-0.5mm}
\subsection{Experimental Settings}

\smallsection{Machines:} We conducted all the experiments on a machine with an AMD Ryzen 9 3900X CPU and 128GB RAM. 

\smallsection{Implementations:}
We implemented all versions of \method using C++ and OpenMP. 

\smallsection{Datasets:} We used the following eleven real-world hypergraphs from five different domains:
\vspace{-0.5mm}
\begin{itemize}[leftmargin=*]
	\itemsep-0.2em 
	\item \textbf{\texttt{co-authorship}} (coauth-DBLP, coauth-geology~\cite{sinha2015overview}, and coauth-history~\cite{sinha2015overview}): A node represents an author. A hyperedge represents all authors of a publication. 
	\item \textbf{\texttt{contact}} (contact-primary~\cite{stehle2011high} and contact-high~\cite{mastrandrea2015contact}): A node represents a person. A hyperedge represents a group interaction among individuals.
	\item \textbf{\texttt{email}} (email-Enron~\cite{klimt2004enron} and email-EU~\cite{leskovec2005graphs,yin2017local}): A node represents an e-mail account. A hyperedge consists of the sender and all receivers of an email.
	\item \textbf{\texttt{tags}} (tags-ubuntu and tags-math): A node represents a tag. A hyperedge represents all tags attached to a post.
	\item \textbf{\texttt{threads}} (threads-ubuntu and threads-math): A node represents a user. A hyperedge groups all users participating in a thread.
\end{itemize}
\vspace{-0.5mm}
These hypergraphs are made public by the authors of \cite{benson2018simplicial}, and in Table~\ref{dataset_table} we provide some statistics of the hypergraphs after removing duplicated hyperedges.
We used $\methodE$ for the \textit{coauth-history} dataset, the \textit{threads-ubuntu} dataset, and all datasets from the \texttt{contact} and \texttt{email} domains. For the other datasets, we used $\methodAW$ with $r=2,000,000$, unless otherwise stated. 
We used a single thread unless otherwise stated.
We computed CPs based on five hypergraphs randomized as described in Section~\ref{sec:concept:motif}.

\vspace{-0.5mm}
\subsection{Q1. Comparison with Random}
\label{sec:exp:random}

We analyze the counts of different \motifs' instances in real and random hypergraphs.
In Table~\ref{absolute_table}, we report the (approximated) count of each \motif $t$'s instances in each real hypergraph with the corresponding count averaged over five random hypergraphs obtained as described in Section~\ref{sec:concept:motif}.
\change{For each \motif $t$, we measure its relative count, which we define as $\frac{\MT- M_{rand}[t]}{\MT+  M_{rand}[t]}.$}
We also rank \motifs by the counts of their instances and examine the difference between the ranks in real and corresponding random hypergraphs.
As seen in the table, the count distributions in real hypergraphs are clearly distinguished from those of random hypergraphs.




\smallsection{\Motifs in Random Hypergraphs:}
We notice  that instances of \motifs $17$ and $18$ appear much more frequently in random hypergraphs than in real hypergraphs from all domains. For example, instances of \motif $17$ appear only about $510$ thousand times in the \textit{tags-math} dataset, while they appear about $500$ million times (about $\mathbf{980 \times}$ more often) in the corresponding randomized hypergraph. In the \textit{threads-math} dataset, instances of \motif $18$ appear about $830$ thousand times, while they appear about $13$ billion times (about $\mathbf{15,660\times}$ more often) in the corresponding randomized hypergraph. 
Instances of \motifs $17$ and $18$ consist of a hyperedge and its two disjoint subsets (see Figure~\ref{motif_three_hyperedges}).



\smallsection{\Motifs in Co-authorship Hypergraphs:}
We observe that instances of \motifs $10$, $11$ and $12$ appear more frequently in all three hypergraphs from the \texttt{co-authorship} domain than in the corresponding random hypergraphs.
Although there are only about $190$ instances of \motif $12$ in the corresponding random hypergraphs, there are about $64$ million such instances (about $\mathbf{337,000\times}$ more instances) in the \textit{coauth-DBLP} dataset.
As seen in Figure~\ref{motif_three_hyperedges}, in instances of \motifs $10$, $11$, and $12$, a hyperedge is overlapped with the two other overlapped hyperedges in three different ways.


\smallsection{\Motifs in Contact Hypergraphs:}
Instances of h-mo\-tifs $9$, $13$, and $14$ are noticeably more common in both \texttt{contact} datasets than in the corresponding random hypergraphs.  
As seen in Figure~\ref{motif_three_hyperedges}, in instances of \motifs $9$, $13$ and $14$, hyperedges are tightly connected and nodes are mainly located in the intersections of all or some hyperedges.

\smallsection{\Motifs in Email Hypergraphs:} 
Both email datasets contain particularly many instances of \motifs $8$ and $10$, compared to the corresponding random hypergraphs.
As seen in Figure~\ref{motif_three_hyperedges}, instances of \motifs $8$ and $10$ consist of three hyperedges one of which contains most nodes.

\begin{figure}[t!]  
	\centering
	\vspace{-5mm}
	\includegraphics[width=0.48\textwidth]{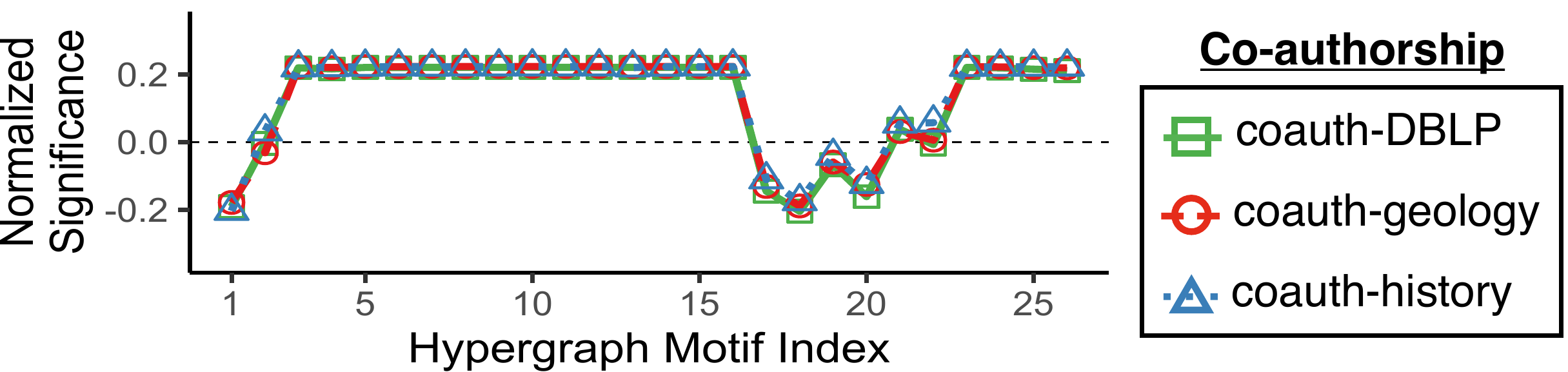}  \\
	\includegraphics[width=0.48\textwidth]{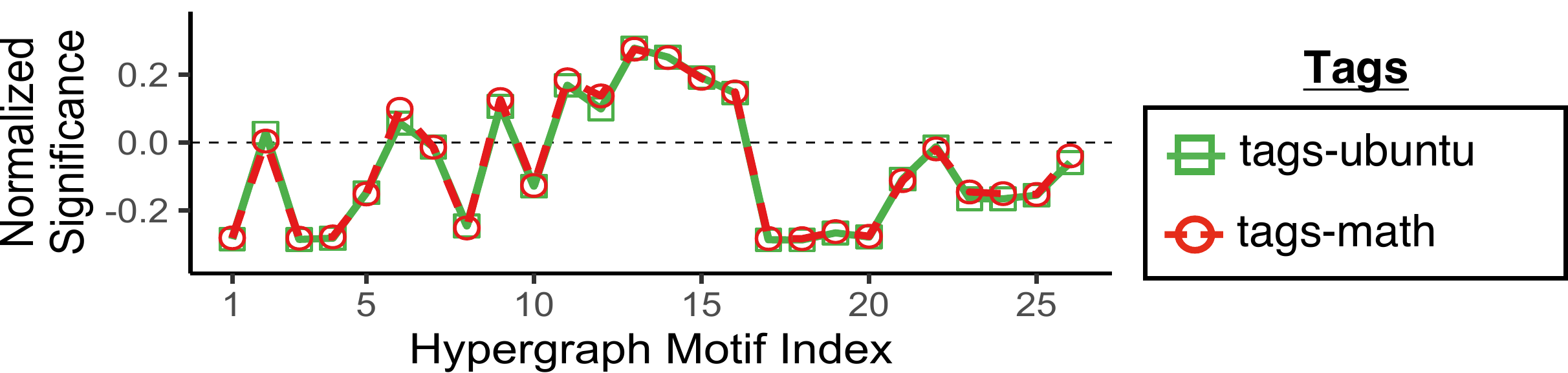} \\
	\includegraphics[width=0.48\textwidth]{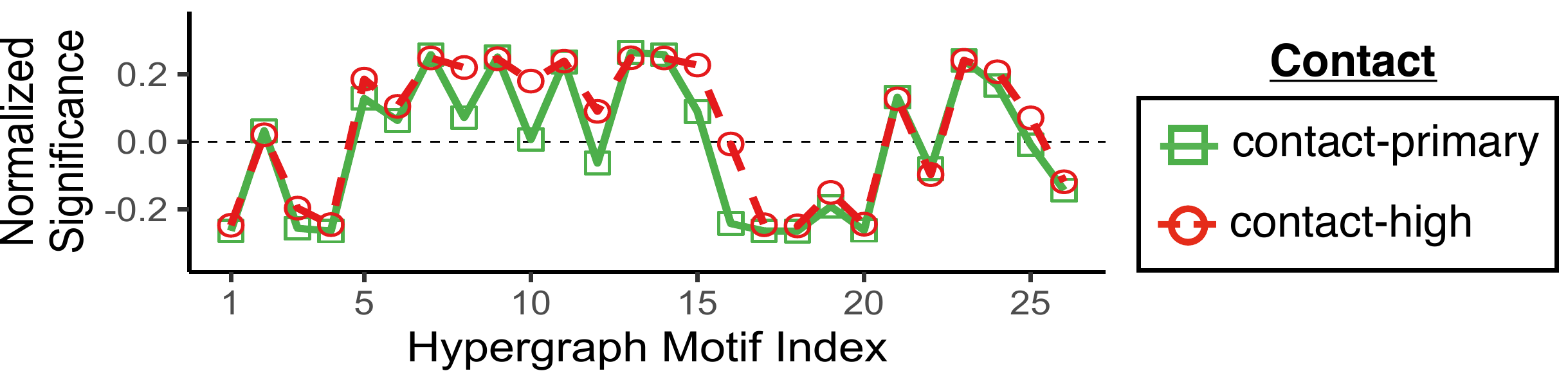}\\
	\includegraphics[width=0.48\textwidth]{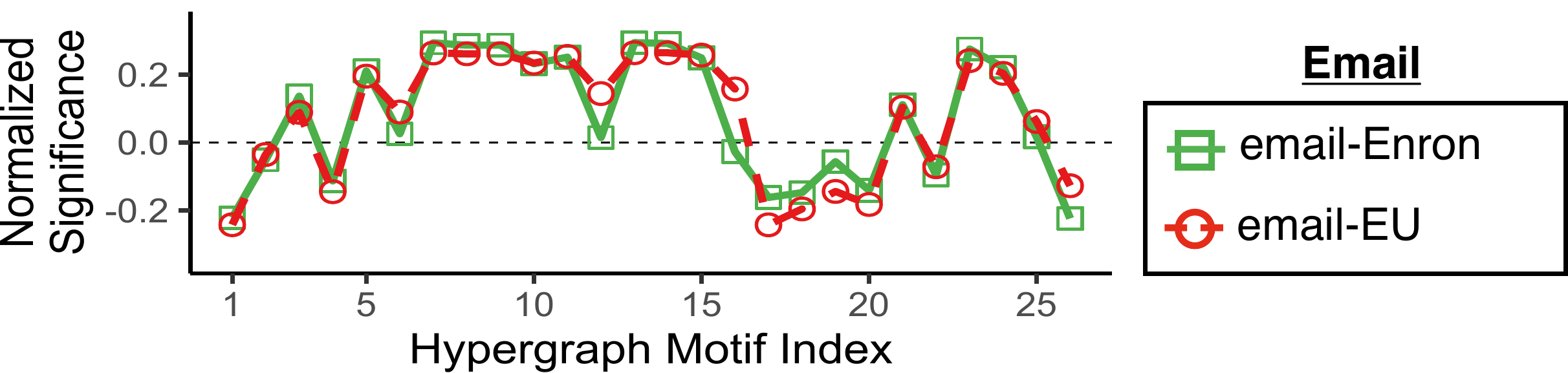}\\
	\includegraphics[width=0.48\textwidth]{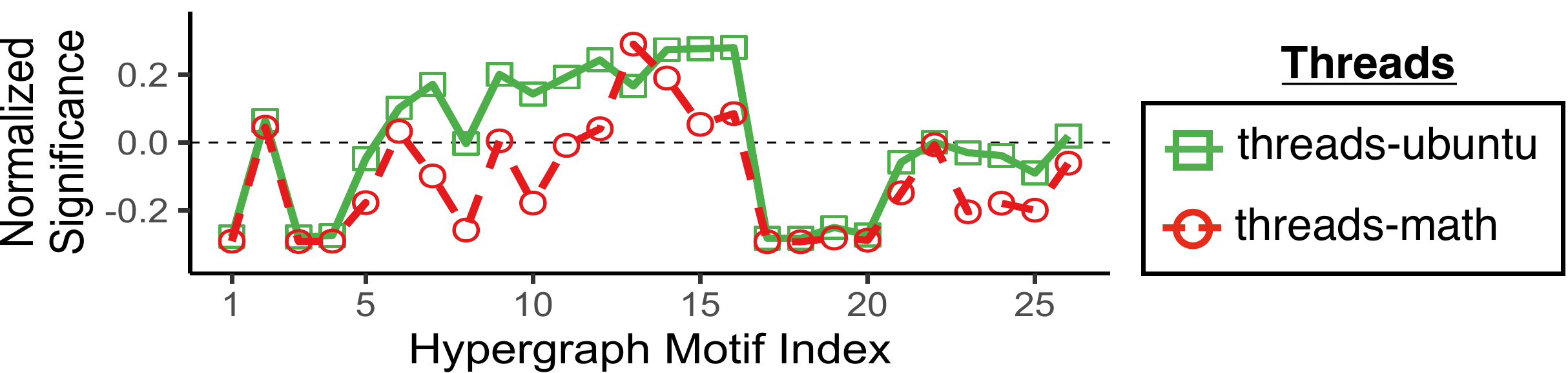}\\
	\vspace{-0.5mm}
	\caption{\label{fig:cp}
		Characteristic profiles (CPs) capture local structural patterns of real-world hypergraphs accurately.
		The CPs are similar within domains but different across domains. 
		Note that the significance of \motif 3 distinguishes 
		the contact hypergraphs from the email hypergraphs.
	}
\end{figure}

\begin{figure}[t!]
	\vspace{-4mm}
	\centering     
	\hspace{-2mm}
	\subfigure[Similarity matrix based on hypergraph motifs]{
		\includegraphics[width=0.557\columnwidth]{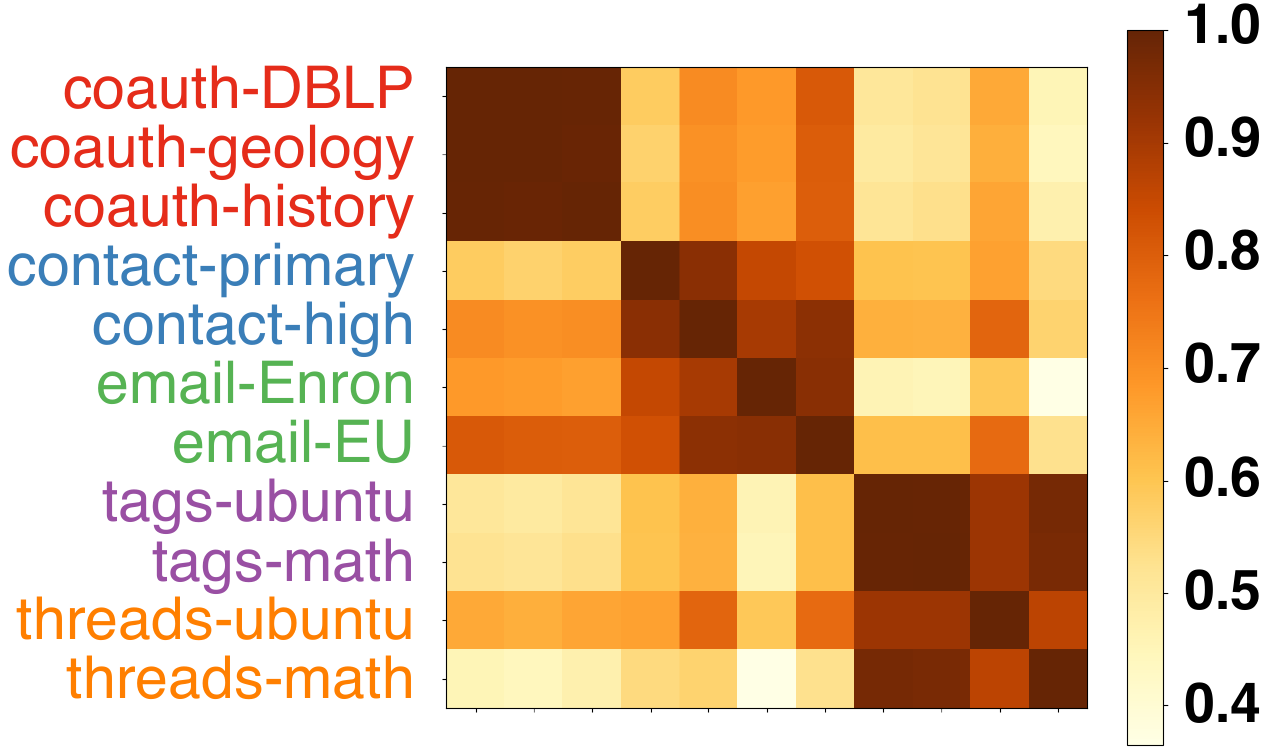}
	}
	\hspace{-2mm}
	\subfigure[Similarity matrix based on network motifs
	]{
		\includegraphics[width=0.403\columnwidth]{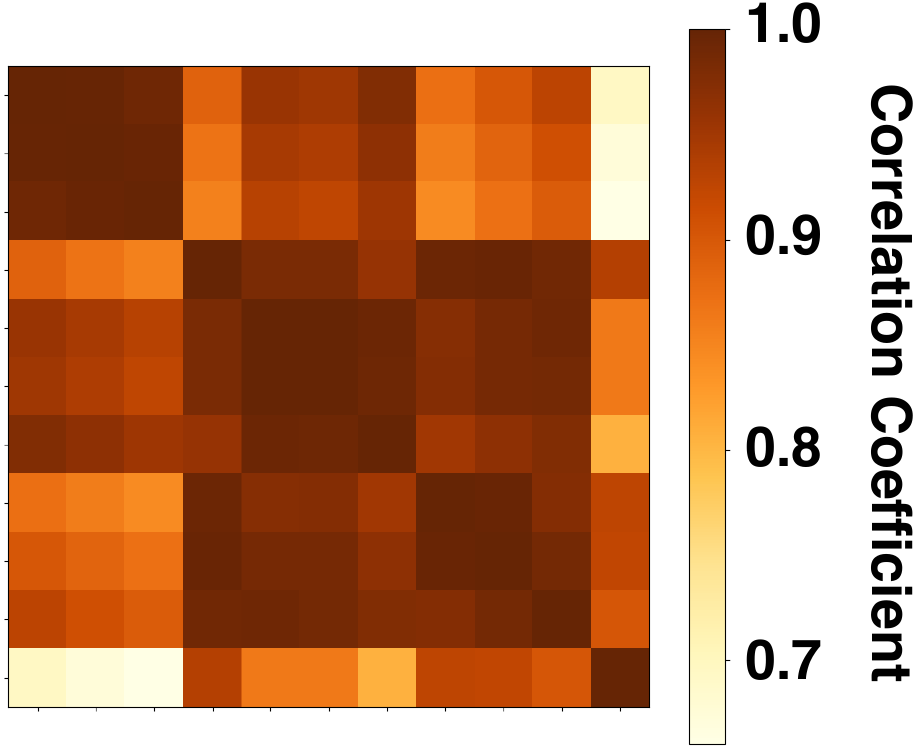}
	} \\
	\vspace{-2mm}
	\caption{\label{heatmap_fig}
		Characteristic profiles (CPs) based on hypergraph motifs (\motifs) capture local structural patterns more accurately than CPs based on network motifs.
		The CPs based on \motifs distinguishes the domains of the real-world hypergraphs better than the CPs based on network motifs. }
\end{figure}

\smallsection{\Motifs in Tags Hypergraphs:}
In addition to instances of \motif $11$, which are common in most real hypergraphs, 
instances of \motif $16$, where all seven regions are not empty (see Figure~\ref{motif_three_hyperedges}), are particularly frequent in both \texttt{tags} datasets than in corresponding random hypergraphs.

\smallsection{\Motifs in Threads Hypergraphs:}
Lastly, in both data sets from the \texttt{threads} domain, instances of \motifs $12$ and $24$ are noticeably more frequent than expected from the corresponding random hypergrpahs.

\change{In Appendix C.1 of \cite{full}, we analyze how the significance of each \motif and the rank difference for it are related to global structural properties of hypergraphs.}

%


\vspace{-0.5mm}
\subsection{Q2. Comparison across Domains}
\label{sec:exp:domain}
We compare the characteristic profiles (CPs) of the real-world hypergraphs.
In Figure~\ref{fig:cp}, we present the CPs (i.e., the significances of the $26$ \motifs) of each hypergraph.
As seen in the figure, hypergraphs from the same domains have similar CPs. 
Specifically, all three hypergraphs from the \texttt{co-authorship} domain share extremely similar CPs, even when the absolute counts of \motifs in them are several orders of magnitude different.
Similarly, the CPs of both hypergraphs from the \texttt{tags} domain are extremely similar.
However, the CPs of the three hypergraphs from the \texttt{co-authorship} domain are clearly distinguished by them of the hypergraphs from the  
\texttt{tags} domain.
While the CPs of the hypergraphs from the \texttt{contact} domain and the CPs of those from the \texttt{email} domain are similar for the most part, they are distinguished by the significance of \motif 3.
These observations confirm that CPs accurately capture local structural patterns of real-world hypergraphs.
\change{In Appendix~\ref{appendix:addExp-motifSignificance} of \cite{full}, we analyze the importance of each \motif in terms its contribution to distinguishing the domains.
}

To further verify the effectiveness of CPs based on \motifs, we compare them with CPs based on network motifs.
Specifically, we represent each hypergraph $G=(V,E)$ as a bipartite graph $G'=(V',E')$ where $V':=V\cup E$ and $E':=\{(v,e)\in V\times E: v\in e\}$, \change{which is also known as \textit{star expansion}~\cite{sun2008hypergraph}}. That is, the two types of nodes in the transformed bipartite graph $G'$ represent the nodes and hyperedges, respectively, in the original hypergraph $G$, and each edge $(v,e)$ in $G'$ indicates that the node $v$ belongs to the hyperedge $e$ in $G$.
Then, we compute the CPs based on the network motifs consisting of $3$ to $5$ nodes, using \cite{bressan2019motivo}.
Lastly, based on both CPs, we compute the similarity matrices (specifically, correlation coefficient matrices) of the real-world hypergraphs.
As seen in Figure~\ref{heatmap_fig}, the domains of the real-world hypergraphs are distinguished more clearly by the CPs based on \motifs than by the CPs based on network motifs.
Numerically, when the CPs based on \motifs are used, the average correlation coefficient is $0.978$ within domains and $0.654$ across domains, and the gap is $0.324$.
However, when the CPs based on network motifs are used, the average correlation coefficient is $0.988$ within domains and $0.919$ across domains, and the gap is just $0.069$.
These results support that \motifs play a key role in capturing local structural patterns of real-world hypergraphs.

%

\begin{figure*}[t]
	\vspace{-4mm}
	\centering     
	\subfigure[\label{fig:dblp:year} Fraction of the instances of each \motif in the coauth-DBLP dataset over time.]{\includegraphics[width=0.71\textwidth]{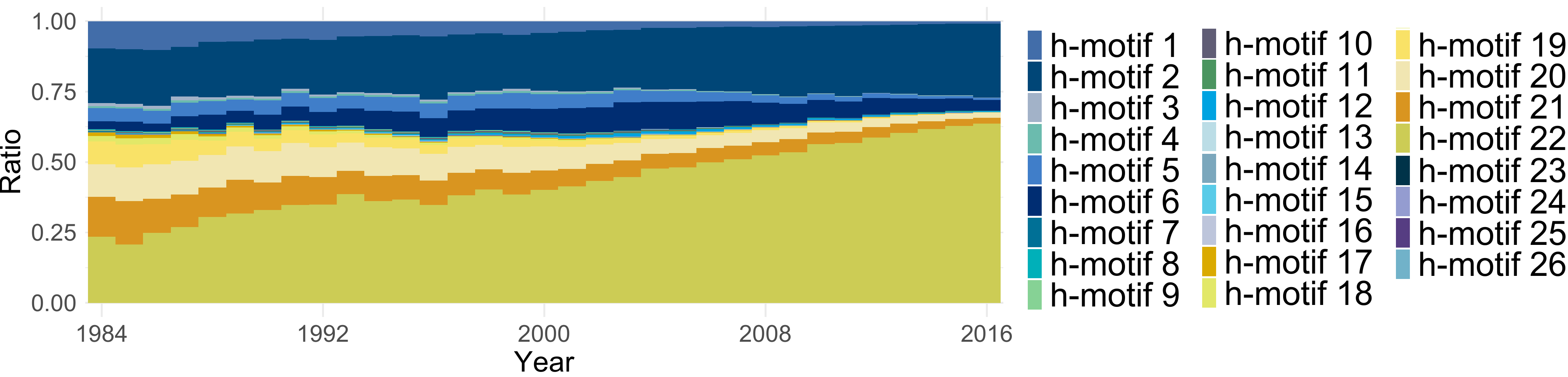}}
	\subfigure[\label{fig:dblp:openclosed} Open and closed \motifs.]{
		\includegraphics[width=0.01\textwidth]{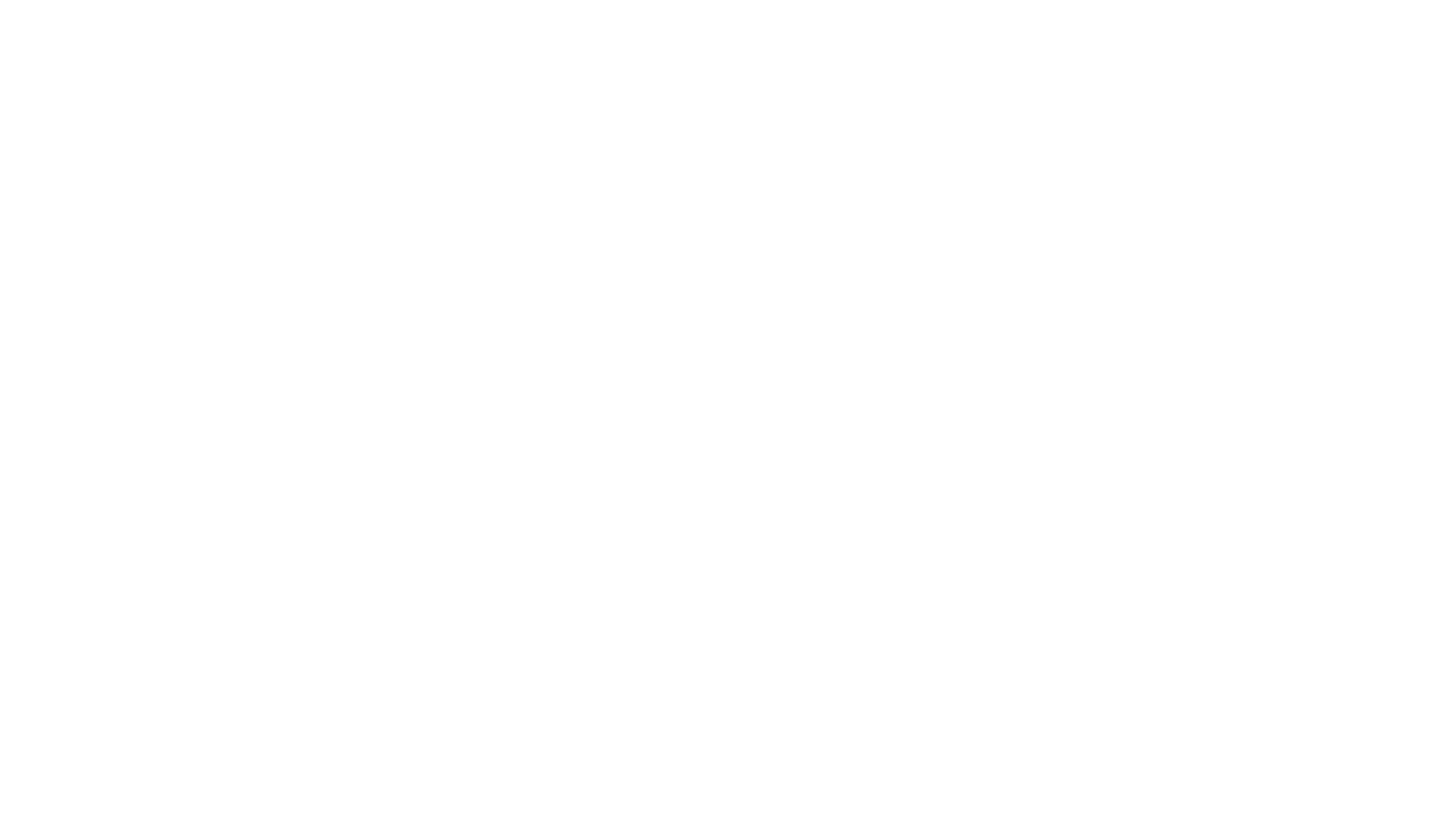}
		\includegraphics[width=0.24\textwidth]{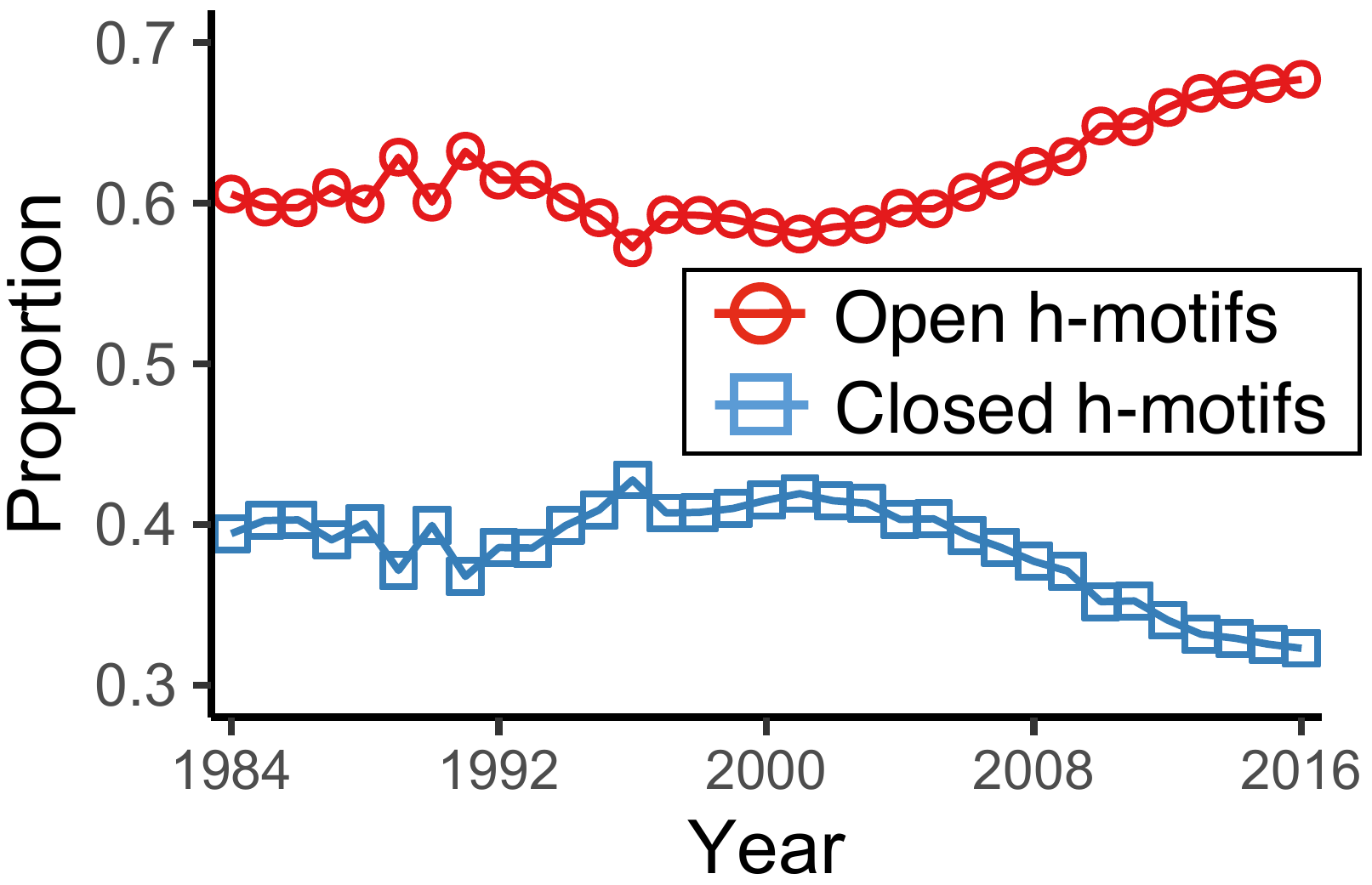}
		\includegraphics[width=0.01\textwidth]{FIG/empty.pdf}} \\
	\vspace{-2mm}
	\caption{\label{fig:dblp_year}\change{Trends in the formation of collaborations are captured by \motifs. (a) The fractions of the instances of \motifs 2 and 22 have increased rapidly.
	(b) The fraction of the instances of open \motifs has increased steadily since 2001. }}
\end{figure*}

\begin{figure*}[t]
	\vspace{-3mm}
	\centering     
	\includegraphics[width=0.09\textwidth]{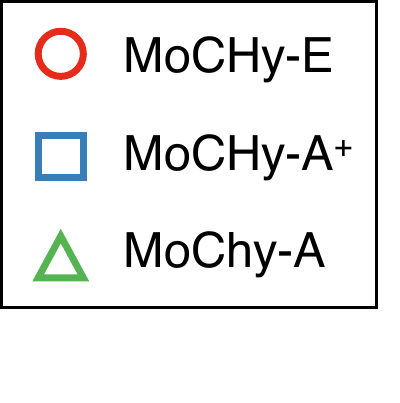} 
	\includegraphics[width=0.146\textwidth]{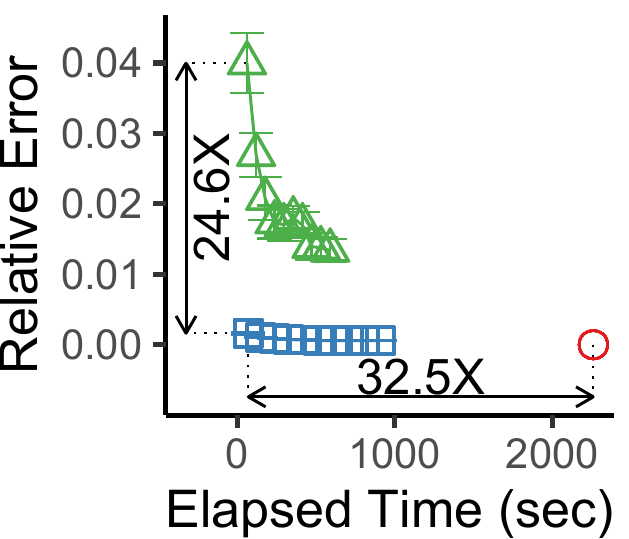}
	\includegraphics[width=0.146\textwidth]{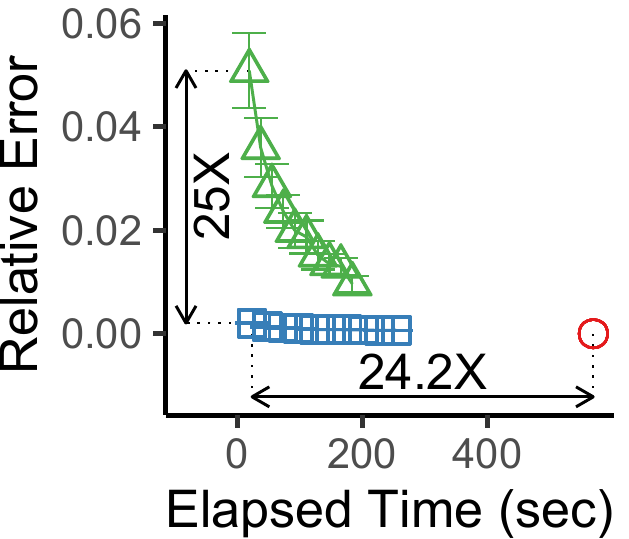}
	\includegraphics[width=0.146\textwidth]{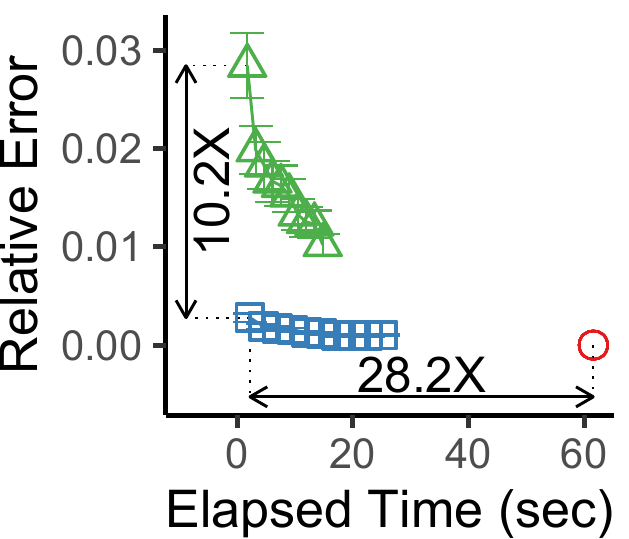}
	\includegraphics[width=0.146\textwidth]{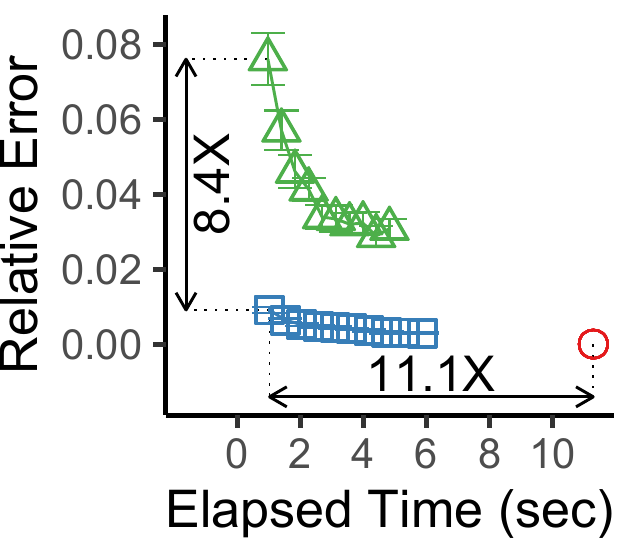}
	\includegraphics[width=0.146\textwidth]{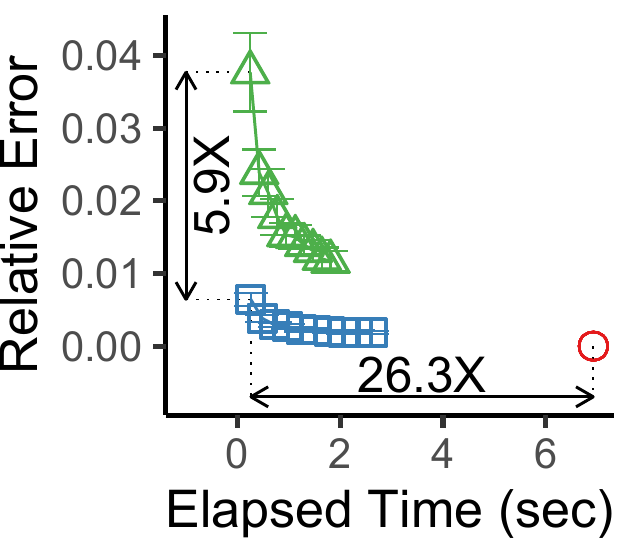}
	\includegraphics[width=0.146\textwidth]{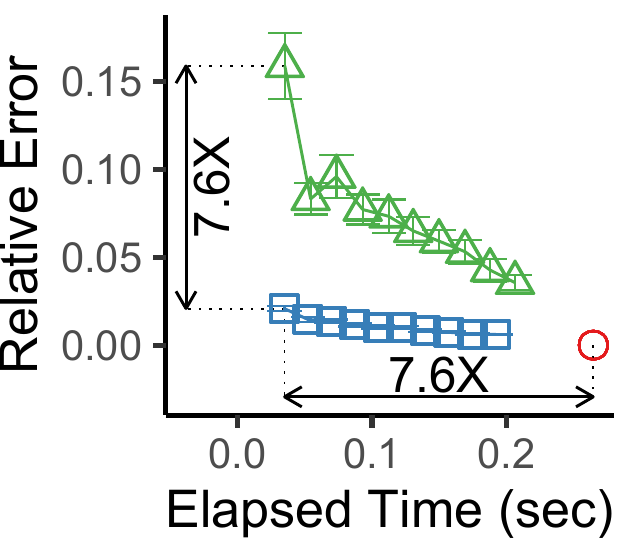} \\
	\begin{tabular}{ccccccc}
		~~~~~~~~~ & (a) threads-ubuntu & (b) email-Eu & (c) contact-primary & (d) coauth-history & (e) contact-high & (f) email-Enron \\
	\end{tabular}
	\caption{\label{fig:tradeoff}
		\methodAWX gives the best trade-off between speed and accuracy. It yields up to $25 \times$ more accurate estimation than \methodAEX, and it is up to $32.5 \times$ faster than \methodEX.
		The error bars indicate $\pm$ $1$ standard error over $20$ trials.
	}
\end{figure*}

\begin{table}[t]
	\vspace{-4mm}
	\begin{center}
		\caption{\label{prediction_table}
				\Motifs give informative features.
				Using them in HM26 or HM7 yields more accurate hyperedge predictions than using the baseline features in HC.}
		\scalebox{0.82}{
			\begin{tabular}{cc|ccc}
				\toprule
				\multicolumn{2}{c|}{} &  \textbf{HM26} &  \textbf{HM7} &  \textbf{HC}\\
				\midrule
				\multirow{2}{*}{ \textbf{Logistic Regression}} & \textbf{ACC$*$} & 0.754 & 0.656 & 0.636\\
				& \textbf{AUC$\dagger$} & 0.813 & 0.693 & 0.691\\
				\midrule
				\multirow{2}{*}{ \textbf{Random Forest}} & \textbf{ACC} & 0.768 & 0.741 & 0.639\\
				& \textbf{AUC} & 0.852 & 0.779 & 0.692\\
				\midrule
				\multirow{2}{*}{ \textbf{Decision Tree}} & \textbf{ACC} & 0.731 & 0.684 & 0.613\\
				& \textbf{AUC} & 0.732 & 0.685 & 0.616\\
				\midrule
				\multirow{2}{*}{ \textbf{K-Nearest Neighbors}} & \textbf{ACC} & 0.694 & 0.689 & 0.640\\
				& \textbf{AUC} & 0.750 & 0.743 & 0.684\\
				\midrule
				\multirow{2}{*}{ \textbf{MLP Classifier}} & \textbf{ACC} & 0.795 & 0.762 & 0.646\\
				& \textbf{AUC} & 0.875 & 0.841 & 0.701\\
				\bottomrule
				\multicolumn{5}{l}{$*$ acuracy, $\dagger$ area under the ROC curve. }
			\end{tabular}
		}
	\end{center}
	\vspace{-2mm}
\end{table}


%
\vspace{-0.5mm}
\subsection{Q3. Observations and Applications}
\label{sec:exp:observations}
We conduct two case studies on the \textit{coauth-DBLP} dataset, which is a \texttt{co-authorship} hypergraph.

\smallsection{Evolution of Co-authorship Hypergraphs:}\label{sec:exp:observations:evolution}  
The data-set contains bibliographic information of computer science publications.	
Using the publications in each year from $1984$ to $2016$, we create $33$ hypergraphs where each node corresponds to an author, and each hyperedge indicates the set of the authors of a publication.
Then, we compute the fraction of the instances of each \motif in each hypergraph to analyze patterns and trends in the formation of collaborations.
 As shown in Figure~\ref{fig:dblp_year}, over the 33 years, the fractions have changed with distinct trends.
 First, as seen in Figure~\ref{fig:dblp:openclosed}, the fraction of the instances of open \motifs has increased steadily since 2001, 
 indicating that collaborations have become less clustered, i.e., the probability that two collaborations intersecting with a collaboration also intersect with each other has decreased. 
 Notably, the fractions of the instances of \motif $2$ (closed) and \motif $22$ (open) have increased rapidly, accounting for most of the instances.

\smallsection{\change{Hyperedge Prediction:}\label{sec:exp:observations:prediction}} 
\change{
As an application of \motifs, we consider the problem of predicting publications (i.e., hyperedges) in $2016$ based on the publications from $2013$ to $2015$.
 As in~\cite{yoon2020much}, we formulate this problem 
 as a binary classification problem where we aim to classify real hyperedges and fake ones.
 To this end, we create fake hyperedges in both training and test sets by replacing some fraction of nodes in each real hyperedge with random nodes (see Appendix E of \cite{full} for details).
 Then, we train five classifiers using the following three different sets of input hyperedge features:
 \vspace{-0.5mm}
 \begin{itemize}[leftmargin=*]
	\itemsep-0.2em 
	\item \textbf{HM26} ($\in\mathbb{R}^{26}$): The number of each \motif's instances that contain each hyperedge.
	\item \textbf{HM7} ($\in\mathbb{R}^{7}$): The seven features with the largest variance among those in \textbf{HM26}. 
	\item \textbf{HC} ($\in\mathbb{R}^{7}$): The mean, maximum, and minimum degree\footnote{\scriptsize \change{The degree of a node $v$ is the number of hyperedges that contain $v$.}} and the mean, maximum, and minimum number of neighbors\footnote{\scriptsize \change{The neighbors of a node $v$ is the nodes that appear in at least one hyperedge together with $v$.}} of the nodes in each hyperedge and its size.
 \end{itemize}
 \vspace{-0.5mm}
We report the accuracy (ACC) and the area under the ROC curve (AUC) in each setting in Table~\ref{prediction_table}. 
Using \textbf{HM7}, which is based on \motifs, yields consistently better predictions than using an equal number of baseline features in \textbf{HC}.
Using \textbf{HM26} yields the best predictions.
These results indicate informative features can be obtained from \motifs.} \\


\vspace{-0.5mm}
\subsection{Q4. Performance of Counting Algorithms}
\label{sec:exp:algo}

We test the speed and accuracy of all versions of \method under various settings.
To this end, we measure elapsed time and relative error defined as
$$\frac{\sum_{t=1}^{26}|\MT-\MBT|}{\sum_{t=1}^{26}\MT} \text{  and  } \frac{\sum_{t=1}^{26}|\MT-\MHT|}{\sum_{t=1}^{26}\MT},$$
for \methodAE and \methodAW, respectively.

\smallsection{Speed and Accuracy:}
In Figure~\ref{fig:tradeoff}, we report the elapsed time and relative error of all versions of \method on the $6$ different datasets where \methodE terminates within a reasonable time.
The numbers of samples in \methodAE and \methodAW are set to $\{2.5\times k: 1\leq k \leq 10\}$ percent of the counts of hyperedges and \hwedges, respectively. \methodAW provides the best trade-off between speed and accuracy. For example, in the \textit{threads-ubuntu} dataset, \methodAW provides $24.6\times$ lower relative error than \methodAE, consistently with our theoretical analysis (see the last paragraph of Section~\ref{sec:method:approx}). Moreover, in the same dataset, \methodAW is $32.5\times$ faster than \methodE with little sacrifice on accuracy.

\smallsection{Effects of the Sample Size on CPs:}
In Figure~\ref{sensitivity_fig}, we report the CPs obtained by \methodAW with different numbers of hyperwedge samples on $3$ datasets.
Even with a smaller number of samples, the CPs are estimated near perfectly.


\smallsection{Parallelization:}
We measure the running times of \methodE and \methodAW with different numbers of threads on the \textit{threads-ubuntu} dataset. As seen in Figure~\ref{fig:par}, both algorithms achieve significant speedups with multiple threads. Specifically, with $8$ threads, \methodE and \methodAW ($r=1M$) achieve speedups of $5.4$ and $6.7$, respectively.

\smallsection{Effects of On-the-fly Computation on Speed:}
We analyze the effects of the on-the-fly computation of projected graphs (discussed in Section~\ref{sec:method:par_fly}) on the speed of \methodAW under different memory budgets for memoization.
To this end, we use the \textit{threads-ubuntu} dataset, and we set the memory budgets so that up to $\{0\%, 0.1\%, 1\%, 10\%, 100\%\}$ of the edges in the projected graph can be memoized.
The results are shown in Figure~\ref{fig:mem}.
As the memory budget increases, \methodAW becomes faster, avoiding repeated computation.
Due to our careful prioritization scheme based on degree, memoizing $1\%$ of the edges achieves speedups of about $2$.


\begin{figure*}[t]
	\centering     
	\vspace{-4.5mm}
	\includegraphics[width=0.155\textwidth]{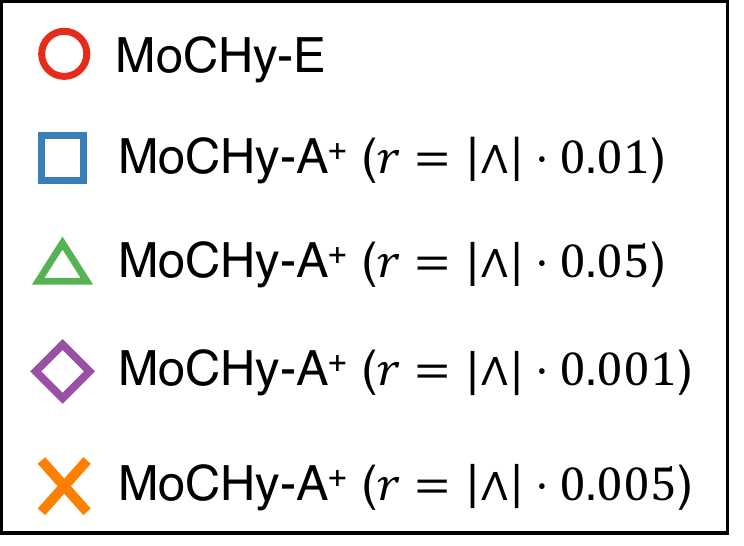}
	\subfigure[email-EU]{\includegraphics[width=0.276\textwidth]{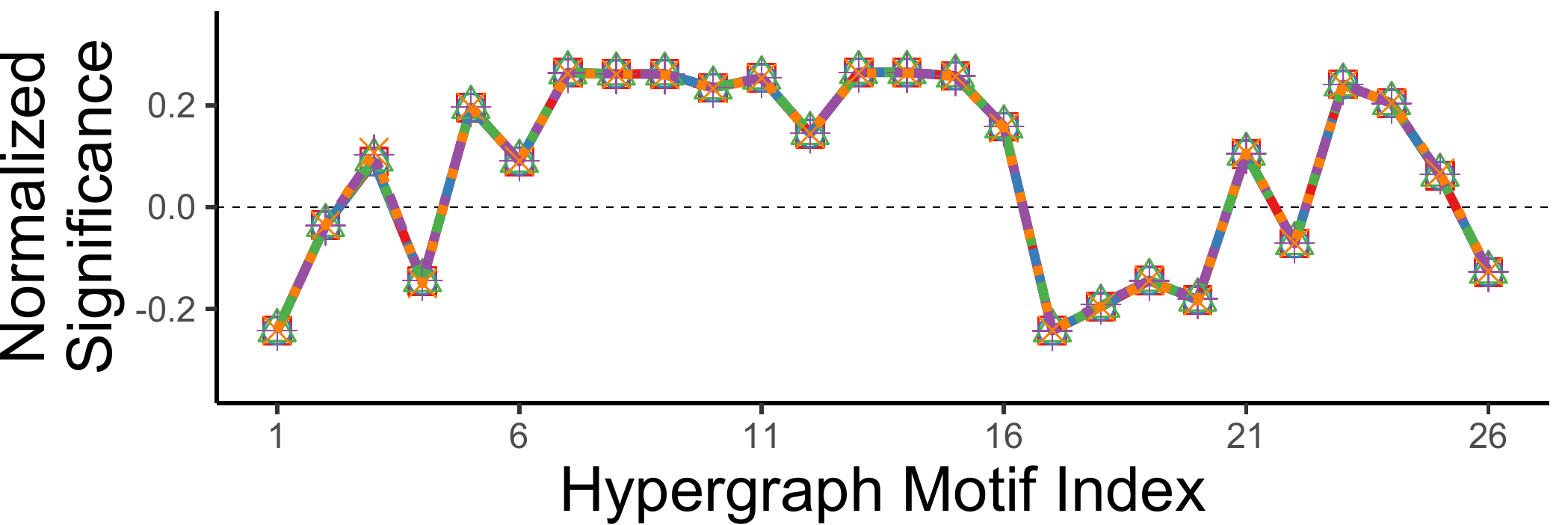}}
	\subfigure[contact-primary]{\includegraphics[width=0.276\textwidth]{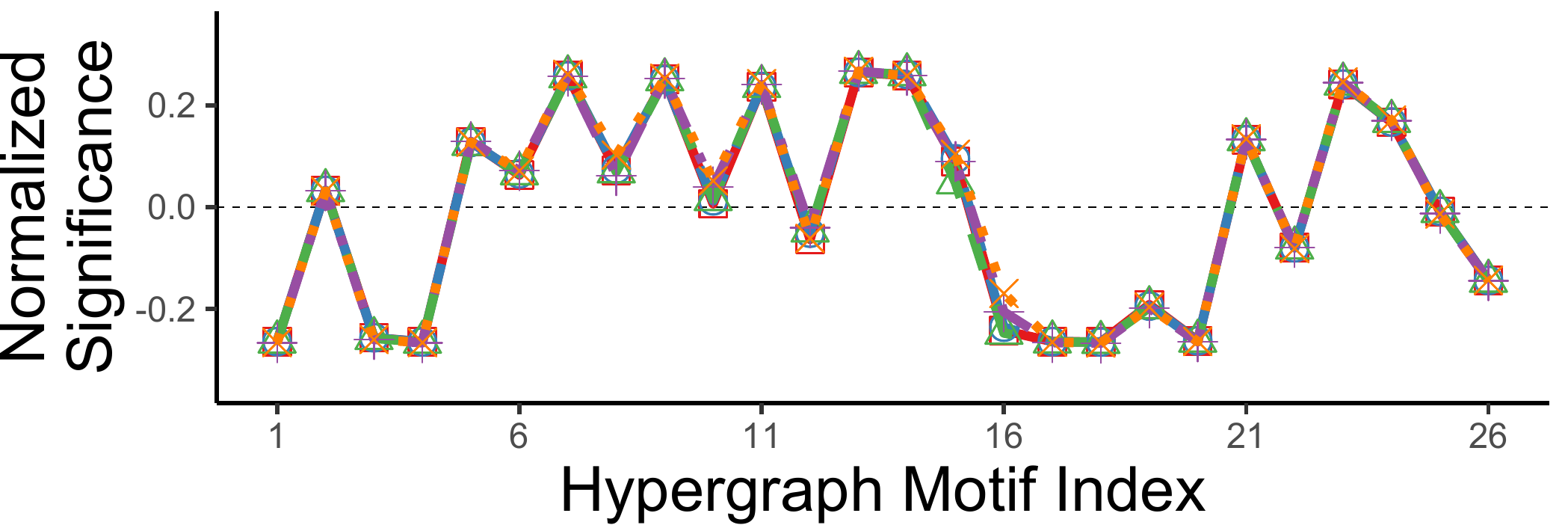}}
	\subfigure[coauth-history]{\label{sensitivity_fig:a}\includegraphics[width=0.276\textwidth]{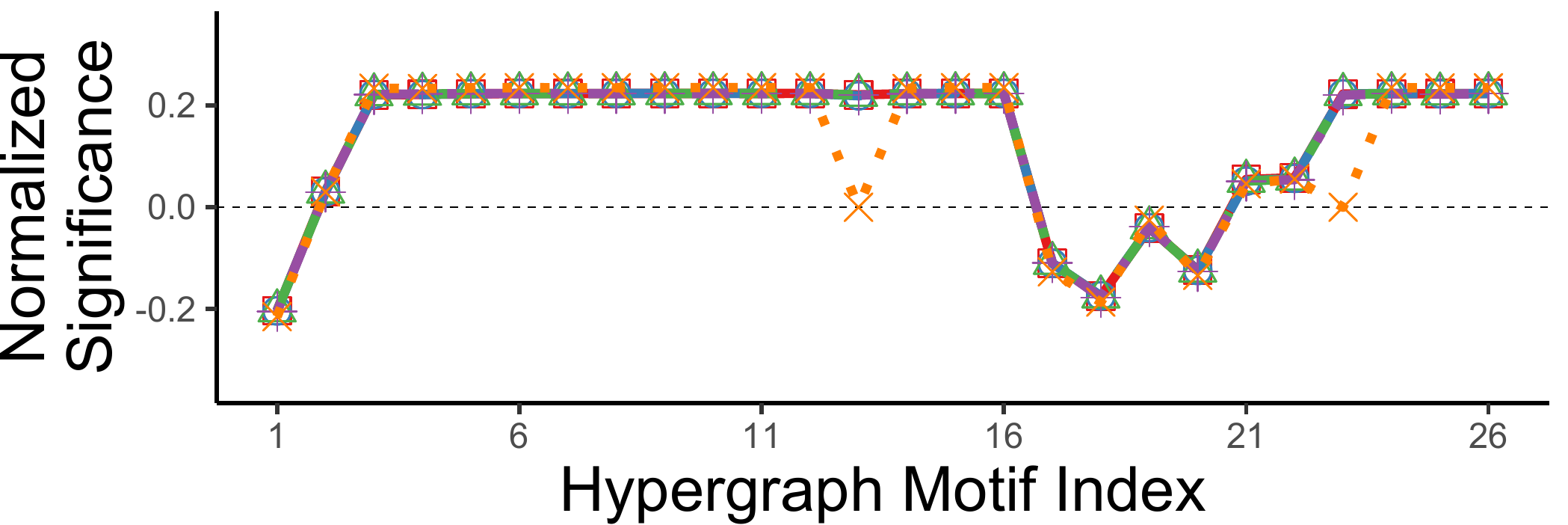}} \\
	\vspace{-2mm}
	\caption{\label{sensitivity_fig}Using \methodAWX, characteristic profiles (CPs) can be estimated accurately from a small number of samples.}
\end{figure*}

\vspace{-0.5mm}
\section{Related Work}
\label{sec:related}

We review prior work on network motifs, algorithms for counting them, and hypergraphs.
While the definition of a network motif varies among studies, here we define it as a connected graph composed by a predefined number of nodes. 

\smallsection{Network Motifs}. Network motifs were proposed as a tool for understanding the underlying design principles and capturing the local structural patterns of graphs \cite{holland1977method,shen2002network,milo2002network}. 
The occurrences of motifs in real-world graphs are significantly different from those in random graphs \cite{milo2002network}, and they vary also depending on the domains of graphs \cite{milo2004superfamilies}.
The concept of network motifs has been extended to various types of graphs, including dynamic \cite{paranjape2017motifs}, bipartite \cite{borgatti1997network}, and heterogeneous \cite{rossi2019heterogeneous} graphs. 
The occurrences of network motifs have been used in a wide range of graph applications: community detection \cite{benson2016higher,yin2017local,li2019edmot,tsourakakis2017scalable}, ranking \cite{zhao2018ranking}, graph embedding \cite{rossi2018higher,yu2019rum}, and graph neural networks \cite{lee2019graph}, to name a few.

\begin{figure}[t]
	\vspace{-0.5mm}
	\centering     
	\includegraphics[width=0.13\textwidth]{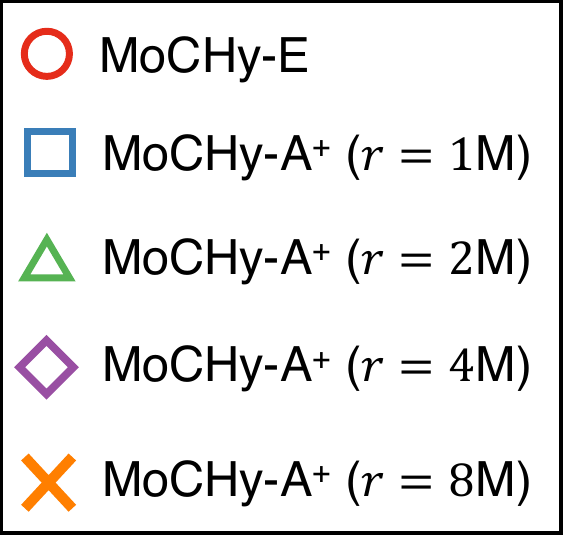} 
	\subfigure[\label{fig:par:time} Elapsed Time]{\includegraphics[width=0.16\textwidth]{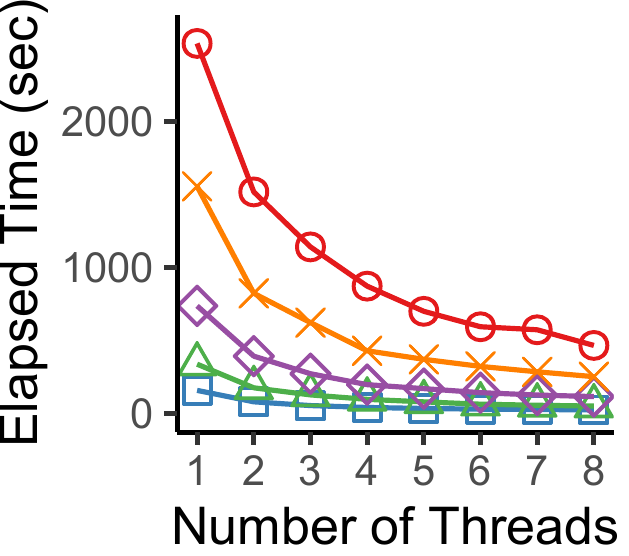}}
	\subfigure[\label{fig:par:speedup} Speedup]{\includegraphics[width=0.16\textwidth]{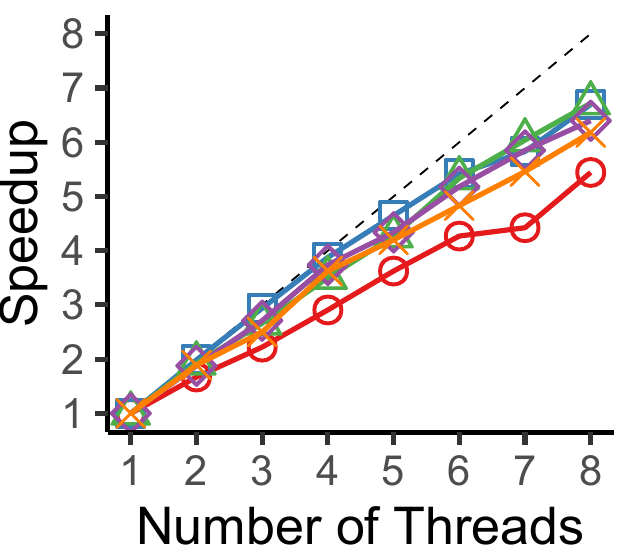}} \\
	\vspace{-2mm}
	\caption{\label{fig:par} Both \methodEX and \methodAWX achieve significant speedups with multiple threads.}
\end{figure}

\smallsection{Algorithms for Network Motif Counting.}
	We focus on algorithms for counting the occurrences of every network motif whose size is fixed or within a certain range \cite{ahmed2015efficient,ahmed2017graphlet,aslay2018mining,bressan2019motivo,chen2016general,han2016waddling,pinar2017escape}, while
	many are for a specific motif (e.g., the clique of size $3$) \cite{ahmed2017sampling,de2016triest,hu2013massive,hu2014efficient,jha2013space,kim2014opt,ko2018turbograph,pagh2012colorful,sanei2018butterfly,shin2017wrs,shin2020fast,tsourakakis2009doulion,wang2019rept,wang2017approximately}.	
	Given a graph,
	they aim to count rapidly and accurately the instances of motifs with $4$ or more nodes, despite the combinatorial explosion of the instances, using the following techniques:
	\vspace{-1.5mm}
	\begin{enumerate}[leftmargin=*]
		\itemsep-0.2em 
		{\setlength\itemindent{5pt}\item[(1)] {\bf Combinatorics:} For exact counting, \cite{ahmed2015efficient,pinar2017escape,paranjape2017motifs} employ combinatorial relations between counts.
		That is, they deduce the counts of the instances of motifs from those of other smaller or equal-size motifs.}
		{\setlength\itemindent{5pt}\item[(2)] {\bf MCMC:} Most approximate algorithms sample motif instances from which they estimate the counts.
		Employed MCMC sampling, \cite{bhuiyan2012guise,chen2016general,han2016waddling,saha2015finding,wang2014efficiently}
		perform a random walk over instances (i.e, connected subgraphs) until it reaches the stationarity
		 to sample an instance from a fixed probability distribution (e.g., uniform).}
		{\setlength\itemindent{5pt}\item[(3)] {\bf Color Coding:} Instead of MCMC, \cite{bressan2017counting,bressan2019motivo} employ color coding \cite{alon1995color}. Specifically, they color each node uniformly at random among $k$ colors, count the number of $k$-trees with $k$ colors rooted at each node, and use them to sample instances from a fixed probability distribution.}
	\end{enumerate}
	\vspace{-2mm}
	In our problem, which focuses on \motifs with only $3$ hyperedges, sampling instances with fixed probabilities is straightforward without (2) or (3), 
	and the combinatorial relations on graphs in (1) are not applicable.
	In algorithmic aspects, we address the computational challenges discussed at the beginning of Section~\ref{sec:method} by answering
	(a) what to precompute (Section~\ref{sec:method:projection}),
	(b) how to leverage it (Sections~\ref{sec:method:exact} and \ref{sec:method:approx}), and 
	(c) how to prioritize it (Sections~\ref{sec:method:par_fly} and \ref{sec:exp:algo}), with formal analyses (Lemma~\ref{lemma:motif:time}; Theorems~\ref{thm:exact:time},~\ref{thm:sampling_ver1:time}, and \ref{thm:sampling_ver2:time}).


\smallsection{Hypergraph}. 
Hypergraphs naturally represent group interactions  
occurring in a wide range of fields, including computer vision \cite{huang2010image,yu2012adaptive}, bioinformatics \cite{hwang2008learning}, circuit design \cite{karypis1999multilevel,ouyang2002multilevel}, social network analysis \cite{li2013link,yang2019revisiting}, and recommender systems \cite{bu2010music,li2013news}.
There also has been considerable attention on machine learning on hypergraphs, including clustering \cite{agarwal2005beyond,amburg2019hypergraph,karypis2000multilevel,zhou2007learning}, classification \cite{jiang2019dynamic,sun2008hypergraph,yu2012adaptive} and hyperedge prediction \cite{benson2018simplicial,yoon2020much,zhang2018beyond}.
Recent studies on real-world hypergraphs revealed interesting patterns commonly observed across domains, including global structural properties (e.g., giant connected components and small diameter) \cite{do2020multi} and temporal patterns regarding arrivals of the same or similar hyperedges \cite{benson2018sequences}.
Notably, Benson et al. \cite{benson2018simplicial} studied how several local features, including edge density, average degree, and probabilities of simplicial closure events for $4$ or less nodes\footnote{\scriptsize \change{The emergence of the first hyperedge that includes a set of nodes each of whose pairs co-appear in previous hyperedges. The configuration of the pairwise co-appearances affects the probability.}}, differ across domains.
Our analysis using \motifs is complementary to these approaches in that it (1) captures local patterns systematically without hand-crafted features, (2) captures static patterns without relying on temporal information, and (3) naturally uses hyperedges with any number of nodes without decomposing them into small ones. 

\begin{figure}[t]
	\vspace{-0.5mm}
	\centering     
	\includegraphics[width=0.13\textwidth]{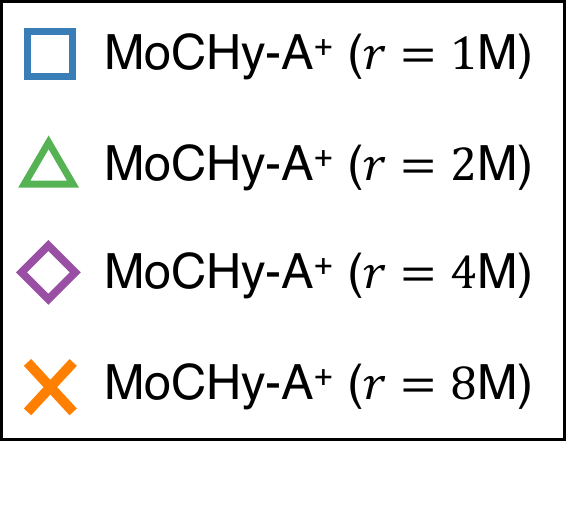} 
	\subfigure[\label{fig:mem:time} Elapsed Time]{\includegraphics[width=0.167\textwidth]{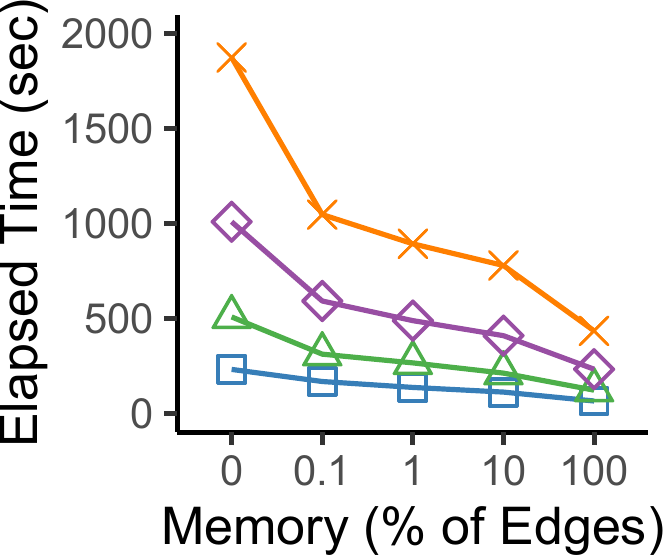}}
	\subfigure[\label{fig:mem:speedup} Speedup]{\includegraphics[width=0.16\textwidth]{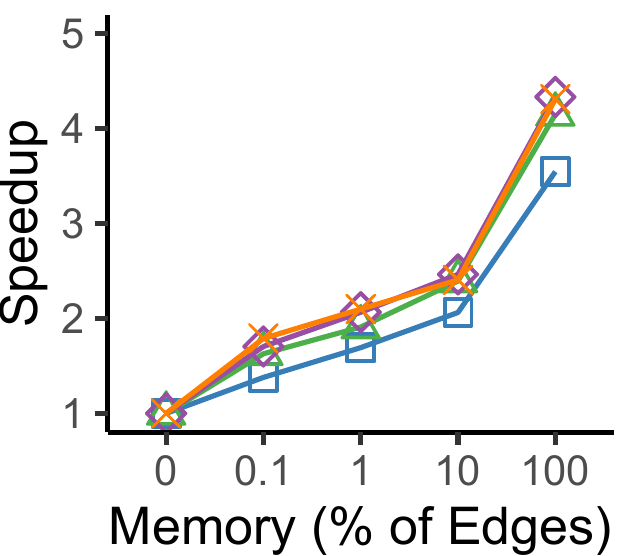}} \\
	\vspace{-2mm}
	\caption{\label{fig:mem}
		Memoizing a small fraction of projected graphs leads to significant speedups of \methodAWX.
	}
\end{figure}

 

\vspace{-0.5mm}
\section{Conclusions}
\label{sec:summary}

In this work, we introduce hypergraph motifs (\motifs), and using them, we investigate the local structures of $11$ real-world hypergraphs from $5$ different domains. We summarize our contributions as follows:
\vspace{-1mm}
\begin{itemize}[leftmargin=*]
	\itemsep-0.3em 
	\item {\bf Novel Concepts:} We define 26 \motifs, which describe connectivity patterns of three connected hyperedges in a unique and exhaustive way, independently of the sizes of hyperedges (Figure~\ref{motif_three_hyperedges}). 
	\item {\bf Fast and Provable Algorithms:} We propose $3$ parallel algorithms for (approximately) counting every \motif's instances, and we theoretically and empirically analyze their speed and accuracy. Both approximate algorithms yield unbiased estimates (Theorems~\ref{thm:sampling_ver1:accuracy} and \ref{thm:sampling_ver2:accuracy}), and especially the advanced one is up to $32\times$ faster than the exact algorithm, with little sacrifice on accuracy (Figure~\ref{fig:tradeoff}).
	\item {\bf Discoveries in $11$ Real-world Hypergraphs:} We confirm the efficacy of \motifs by showing that local structural patterns captured by them are similar within domains but different across domains (Figures \ref{fig:cp} and \ref{heatmap_fig}).
\end{itemize}
\vspace{-2mm}
\noindent{\bf Reproducibility:} The code and datasets used in this work are available at \url{https://github.com/geonlee0325/MoCHy}.


Future directions include  extending \motifs to rich hypergraphs, such as temporal or heterogeneous hypergraphs, and  incorporating \motifs into various tasks, such as hypergraph embedding, ranking, and clustering.

\begin{figure*}
	\vspace{-2mm}
\end{figure*}

{\small 
\smallsection{Acknowledgements.}
This work was supported by National Research Foundation of Korea (NRF) grant funded by the Korea government (MSIT) (No. NRF-2020R1C1C1008296).
This work was also supported by Institute of Information \& Communications Technology Planning \& Evaluation (IITP) grant funded by the Korea government (MSIT) (No. 2019-0-00075, Artificial Intelligence Graduate School Program (KAIST)) and supported by the National Supercomputing Center with supercomputing
resources including technical support (KSC-2020-INO-00\-04).}

\appendix

\vspace{-1mm}
\section{Proof of Theorem~2}
\label{sampling_ver1:proof}
We let $X_{ij}[t]$ be a random variable indicating whether the $i$-th sampled hyperedge (in line~\ref{sampling_ver1:sample} of Algorithm~\ref{sampling_ver1}) is included in the $j$-th instance of \motif $t$ or not. That is, $X_{ij}[t]=1$ if the hyperedge is included in the instance, and $X_{ij}[t]=0$ otherwise. We let $\mBT$ be the number of times that \motif $t$'s instances are counted while processing $s$ sampled hyperedges. That is,
\vspace{-1.5mm}
\begin{equation} 
\mBT := \sum_{i=1}^{s} \sum_{j=1}^{\MT}X_{ij}[t]. \label{eq:mbt}
\vspace{-1.5mm}
\end{equation}
Then, by lines~\ref{sampling_ver1:scale:start}-\ref{sampling_ver1:scale:end} of Algorithm~\ref{sampling_ver1}, 
\vspace{-1mm}
\begin{equation}
\MBT=\mBT\cdot \tfrac{|E|}{3s}. \label{eq:mbt:scale}
\vspace{-1mm}
\end{equation}

\smallsection{Proof of the Bias of $\MBT$ (Eq.~\eqref{sampling_ver1:bias}):}
	Since each \motif instance contains three hyperedges, the probability that each $i$-th sampled hyperedge is contained in each $j$-th instance of \motif $t$ is 
	\vspace{-1.5mm}
	\begin{equation} 
	P[X_{ij}[t]=1] = \mathbb{E}[X_{ij}[t]]=\tfrac{3}{|E|}. \label{eq:Xij:exp}
	\vspace{-0.5mm}
	\end{equation}
	From linearity of expectation,
	\vspace{-0.5mm}
	\begin{equation*} 
	\mathbb{E}[\mBT] = \sum_{i=1}^{s} \sum_{j=1}^{\MT}\mathbb{E}[X_{ij}[t]]= \sum_{i=1}^{s} \sum_{j=1}^{\MT}\frac{3}{|E|} = \frac{3s\cdot \MT}{|E|}. 
	\vspace{-0.5mm}
	\end{equation*}
	Then, by Eq.~\eqref{eq:mbt:scale}, $\mathbb{E}[\MBT] = \tfrac{|E|}{3s} \cdot \mathbb{E}[\mBT] = \MT$. \hfill $\qed$ \\

\vspace{-1.5mm}
\smallsection{Proof of the Variance of $\MBT$ (Eq.~\eqref{sampling_ver1:variance}):} From Eq.~\eqref{eq:Xij:exp} and $X_{ij}[t]=X_{ij}[t]^2$, the variance of $X_{ij}[t]$ is
	\vspace{-0.5mm}
	\begin{equation}
		\mathbb{V}\mathrm{ar}[X_{ij}[t]] = \mathbb{E}[X_{ij}[t]^2] - \mathbb{E}[X_{ij}[t]]^2 = \frac{3}{|E|} - \frac{9}{|E|^2}. \label{eq:Xij:var}
		\vspace{-1mm}
	\end{equation}
	Consider the covariance between $X_{ij}[t]$ and $X_{i'j'}[t]$. 
	If $i = i'$, then from Eq.~\eqref{eq:Xij:exp},
	\vspace{-1mm}
	\begin{align}
			& \mathbb{C}\mathrm{ov}(X_{ij}[t], X_{i'j'}[t]) = \mathbb{E}[X_{ij}[t]\cdot X_{ij'}[t]]-\mathbb{E}[X_{ij}[t]]\mathbb{E}[X_{ij'}[t]] \nonumber\\
			& = P[X_{ij}[t]=1, X_{ij'}[t]=1] -\mathbb{E}[X_{ij}[t]]\mathbb{E}[X_{ij'}[t]] \nonumber\\
			& = P[X_{ij}[t]=1]\cdot P[X_{ij'}[t]=1|X_{ij}[t]=1] \nonumber\\
			& \hspace{10pt} - \mathbb{E}[X_{ij}[t]]\mathbb{E}[X_{ij'}[t]] \nonumber\\
			& = \frac{3}{|E|}\cdot \frac{l_{jj'}}{3} - \frac{9}{|E|^2} = \frac{l_{jj'}}{|E|} - \frac{9}{|E|^2}, \label{eq:Xij:cov}
			\vspace{-1mm}
	\end{align}
	where $l_{jj'}$ is the number of hyperedges that the $j$-th and $j'$-th instances share.	
	However, since hyperedges are sampled independently (specifically, uniformly at random with replacement), if $i \neq i'$, then $\mathbb{C}\mathrm{ov}(X_{ij}[t], X_{i'j'}[t])=0$. 
	This observation, Eq.~\eqref{eq:mbt}, Eq.~\eqref{eq:Xij:var}, and Eq.~\eqref{eq:Xij:cov} imply 
	\vspace{-1mm}
	\begin{align*}
			&\mathbb{V}\mathrm{ar}[\mBT] = \mathbb{V}\mathrm{ar}[\sum_{i=1}^{s} \sum_{j=1}^{\MT}X_{ij}[t]] \nonumber\\
			& = \sum_{i=1}^{s} \sum_{j=1}^{\MT} \mathbb{V}\mathrm{ar}[X_{ij}[t]] + \sum_{i=1}^{s} \sum_{j \neq j'}\mathbb{C}\mathrm{ov}(X_{ij}[t], X_{ij'}[t]) \nonumber\\
			&=s \cdot \MT \cdot  (\frac{3}{|E|} - \frac{9}{|E|^2}) + s \sum_{l=0}^{2} p_l[t] (\frac{l}{|E|} - \frac{9}{|E|^2}),
	\vspace{-1mm}
	\end{align*}
	\vspace{-0.5mm}
	where $p_l[t]$ is the number of pairs of \motif $t$'s instances sharing $l$ hyperedges.
	This and Eq.~\eqref{eq:mbt:scale} imply Eq.~\eqref{sampling_ver1:variance}. $\qed$

\vspace{-2mm}
\section{Proof of Theorem 4}
\label{sampling_ver2:proof}
A random variable $Y_{ij}[t]$ denotes whether the $i$-th sampled \hwedge (in line~\ref{sampling_ver2:sample} of Algorithm~\ref{sampling_ver2}) is included in the $j$-th instance of \motif $t$.
That is, $Y_{ij}[t]=1$ if the sampled \hwedge is included in the instance, and $Y_{ij}[t]=0$ otherwise.
We let $\mHT$ be the number of times that \motif $t$' instances are counted while processing $r$ sampled \hwedges. That is,
\vspace{-2mm}
\begin{equation}
\mHT := \sum_{i=1}^{r} \sum_{j=1}^{\MT}Y_{ij}[t] \label{eq:mht}
\end{equation}
\vspace{-1mm}
We use $w[t]$ to denote the number of \hwedges included in each instance of \motif $t$.
That is,  
\vspace{-0.5mm}
\begin{equation}
w[t]:=\begin{cases}
2 & \text{if \motif $t$ is open,} \\
3 & \text{if \motif $t$ is closed.}
\end{cases}
\vspace{-0.5mm}
\end{equation}

Then, by lines~\ref{sampling_ver2:rescale:start}-\ref{sampling_ver2:rescale:end} of Algorithm~\ref{sampling_ver2}, 
\vspace{-0.5mm}
\begin{equation}
\MHT=
\mHT\cdot \tfrac{1}{w[t]}\cdot \tfrac{|\wedge|}{r}. \label{eq:mht:scale}
\vspace{-1mm}
\end{equation}


\smallsection{Proof of the Bias of $\MHT$ (Eq.~\eqref{sampling_ver2:bias}):} Since each instance of \motif $t$ contains $w[t]$ hyperwedges, the probability that each $i$-th sampled \hwedge is contained in each $j$-th instance of \motif $t$ is 
	\vspace{-0.5mm}
	\begin{equation}
	P[Y_{ij}[t]=1]=\mathbb{E}[Y_{ij}[t]]=\tfrac{w[t]}{|\wedge|}. \label{eq:Yij:exp}
	\vspace{-1mm}
	\end{equation}
	From linearity of expectation,
	\vspace{-0.5mm}
	\begin{equation*} 
		\mathbb{E}[\mHT] = \sum_{i=1}^{r} \sum_{j=1}^{\MT}\mathbb{E}[Y_{ij}[t]]= \sum_{i=1}^{r} \sum_{j=1}^{\MT}\frac{w[t]}{|\wedge|} = \frac{w[t]\cdot r\cdot \MT}{|\wedge|}.
	\vspace{-0.5mm}
	\end{equation*}
	Then, by Eq.~\eqref{eq:mht:scale}, $\mathbb{E}[\MHT] =\mathbb{E}[\mHT]\cdot\tfrac{1}{w[t]}\cdot \tfrac{|\wedge|}{r}=\MT$. $\qed$ \\

\smallsection{Proof of the Variance of $\MHT$ (Eq.~\eqref{sampling_ver2:variance:closed} and Eq.~\eqref{sampling_ver2:variance:open}:} \\
	From Eq.~\eqref{eq:Yij:exp} and $Y_{ij}[t]=Y_{ij}[t]^2$, the variance of each random variable $Y_{ij}[t]$ is
	\vspace{-1.5mm}
	\begin{equation}
	\mathbb{V}\mathrm{ar}[Y_{ij}[t]] = \mathbb{E}[Y_{ij}[t]^2] - \mathbb{E}[Y_{ij}[t]]^2 = \frac{w[t]}{|\wedge|} - \frac{w[t]^2}{|\wedge|^2}. \label{eq:Yij:var}
	\vspace{-1mm}
	\end{equation}
	Consider the covariance between $Y_{ij}[t]$ and $Y_{i'j'}[t]$.
	If $i = i'$, then from Eq.~\eqref{eq:Yij:exp},
	\vspace{-0.5mm}
	\begin{align}
	& \mathbb{C}\mathrm{ov}(Y_{ij}[t], Y_{i'j'}[t])  = \mathbb{E}[Y_{ij}[t]\cdot Y_{ij'}[t]]-\mathbb{E}[Y_{ij}[t]]\mathbb{E}[Y_{ij'}[t]] \nonumber\\
	& = P[Y_{ij}[t]=1, Y_{ij'}[t]=1] -\mathbb{E}[Y_{ij}[t]]\mathbb{E}[Y_{ij'}[t]] \nonumber\\
	& = P[Y_{ij}[t]=1]\cdot P[Y_{ij'}[t]=1|Y_{ij}[t]=1] -\mathbb{E}[Y_{ij}[t]]\mathbb{E}[Y_{ij'}[t]] \nonumber\\
	& = \frac{w[t]}{|\wedge|}\cdot \frac{n_{jj'}}{w[t]} - \frac{w[t]^2}{|\wedge|^2} = \frac{n_{jj'}}{|\wedge|} - \frac{w[t]^2}{|\wedge|^2} \label{eq:Yij:cov},
	\vspace{-0.5mm}
	\end{align}
	where $n_{jj'}$ is the number of \hwedges that the $j$-th and $j'$-th instances share.	
	However, 	
	since \hwedges are sampled independently (specifically, uniformly at random with replacement), if $i\neq i'$, then $\mathbb{C}\mathrm{ov}(Y_{ij}[t], Y_{i'j'}[t])=0$. This observation, Eq.~\eqref{eq:mht}, Eq.~\eqref{eq:Yij:var}, and Eq.~\eqref{eq:Yij:cov} imply 
	\vspace{-2mm}
	\begin{align*}
		&\mathbb{V}\mathrm{ar}[\mHT] = \mathbb{V}\mathrm{ar}[\sum_{i=1}^{r} \sum_{j=1}^{\MT}Y_{ij}[t]] \nonumber\\
		& = \sum_{i=1}^{r} \sum_{j=1}^{\MT} \mathbb{V}\mathrm{ar}[Y_{ij}[t]] + \sum_{i=1}^{r} \sum_{j \neq j'}\mathbb{C}\mathrm{ov}(Y_{ij}[t], Y_{ij'}[t]) \nonumber\\
		& = r \cdot \MT \cdot(\frac{w[t]}{|\wedge|} - \frac{w[t]^2}{|\wedge|^2}) + r \sum_{n=0}^{1}q_n[t] \cdot(\frac{n}{|\wedge|} - \frac{w[t]^2}{|\wedge|^2}), \label{eq:mht:var}
		\vspace{-1mm}
	\end{align*}
	where $q_n[t]$ is the number of pairs of \motif $t$'s instances that share $n$ \hwedges. 
	This and Eq.~\eqref{eq:mht:scale} imply Eq.~\eqref{sampling_ver2:variance:closed} and Eq.~\eqref{sampling_ver2:variance:open}.	\hfill $\qed$


\balance
\bibliographystyle{abbrv}
\bibliography{main}



\end{document}